\documentclass[runningheads]{llncs}
\usepackage{graphicx}
\usepackage{appendix}
\usepackage{amsmath}
\usepackage{amssymb}
\usepackage{amsfonts}
\usepackage{hyperref}
\usepackage{tikz-cd}
\usepackage{amsthm}
\usepackage{algorithmic}
\usepackage{graphicx}
\usepackage[ruled,vlined]{algorithm2e}
\usepackage{tabu}
\theoremstyle{definition}
\usepackage{mathtools}
\newcommand{\defeq}{\vcentcolon=}

\newtheorem{defn}{Definition}

\newcommand{\vol}{\text{Vol}}
\newcommand{\norm}{\text{Norm}}

\newcommand{\trace}{\text{Trace}}

\titlerunning{Subfield Algorithms for Ideal- and Module-SVP}
\begin{document}
\title{Subfield Algorithms for Ideal- and Module-SVP Based on the Decomposition Group}
%
%
\author{Christian Porter \inst{1} \and Andrew Mendelsohn \inst{2} \and  Cong Ling \inst{3}}
\authorrunning{C. Porter et al.}
%
\institute{Department of EEE, Imperial College London, London, SW7 2AZ, United Kingdom.\\ \href{c.porter17@imperial.ac.uk}{c.porter17@imperial.ac.uk} \and Department of EEE, Imperial College London, London, SW7 2AZ, United Kingdom. \\ \href{andrew.mendelsohn18@imperial.ac.uk}{andrew.mendelsohn18@imperial.ac.uk} \and Department of EEE, Imperial College London, London, SW7 2AZ, United Kingdom. \\ \href{c.ling@imperial.ac.uk}{c.ling@imperial.ac.uk}}
\maketitle              
\begin{abstract}
Whilst lattice-based cryptosystems are believed to be resistant to quantum attack, they are often forced to pay for that security with inefficiencies in implementation. This problem is overcome by ring- and module-based schemes such as Ring-LWE or Module-LWE, whose keysize can be reduced by exploiting its algebraic structure, allowing for faster computations. Many rings may be chosen to define such cryptoschemes, but cyclotomic rings, due to their cyclic nature allowing for easy multiplication, are the community standard. However, there is still much uncertainty as to whether this structure may be exploited to an adversary's benefit. In this paper, we show that the decomposition group of a cyclotomic ring of arbitrary conductor can be utilised to significantly decrease the dimension of the ideal (or module) lattice required to solve a given instance of SVP. Moreover, we show that there exist a large number of rational primes for which, if the prime ideal factors of an ideal lie over primes of this form, give rise to an ``easy'' instance of SVP.

It is important to note that the work on ideal SVP does not break Ring-LWE, since its security reduction is from worst case ideal SVP to average case Ring-LWE, and is one way.

\end{abstract}
\section{Introduction}
Cryptosystems based on lattices are one of the leading alternatives to RSA and ECC that are conjectured to be resistant to quantum attacks. Considered to be the genesis of the study of lattice cryptosystems. In 1996, Ajtai constructed a one-way function, and proved the average-case security related to the worst-case complexity of lattice problems \cite{ajtai}. Later, in 2005, Regev proposed the computational problem known as ``Learning With Errors'' (LWE) and showed that LWE is as hard to solve as several worst-case lattice problems \cite{LWE}. Whilst these problems are believed to be difficult to crack even given access to a quantum computer, their main drawback is their impracticality to implement in cryptosystems due to the large key sizes required to define them.
\\ \indent For this reason, lattices with algebraic structure are often favoured to define cryptosystems over conventional lattices. In particular, cryptosystems often employ the use cyclotomic polynomials as their cyclic nature allows for less cumbersome computations. Such lattices come in two main varieties: ideal lattices, whose structures are formed entirely by embedding an ideal of a number field into real or complex space. Module lattices, which are free modules defined over an algebraic ring and can be thought of as a compromise between classical and ideal lattices. 
\\
Perhaps the most well-known cryptosystem based on algebraic structure is the NTRU cryptosystem. Developed in 1996 by Hoffstein, Pipher and Silverman \cite{NTRU}, the NTRU cryptosystem uses elements of the convolution ring $\mathbb{Z}[x]/(x^p-1)$ and offers efficient encryption and decryption of messages, making it one of the most popular lattice-based cryptosystems even to this day. In \cite{securentru}, noting that the ring $\mathbb{Z}[x]/(x^p-1)$ could be deemed insecure due to the fact that $x^p-1$ is reducible, Stehl\'e and Steinfeld updated NTRU to incorporate a cyclotomic ring in place of the aforementioned ring. The computational problem involved in breaking NTRU can be thought of as a rank 2 module problem over a cyclotomic ring. In the same vein of ring-based cryptosystems, an algebraic variant of LWE called Ring-LWE (RLWE) was developed by Lyubashevsky, Peikert and Regev in 2010 \cite{RLWE}. It has been shown that the security of this scheme relies heavily on the hardness of ideal lattice reduction \cite{RLWE-hardness}. Moreover, the work by Ajtai was also generalised to the ring case by Micciancio in 2004 \cite{knapsack}. Using an arbitrary ring in place of a classical lattice, he managed to show that obtaining a solution to the ring-based alternative to the knapsack problem on the average was at least as hard as the worst-case instance of various approximation problems over cyclic lattices, even for rings with relatively small degree over $\mathbb{Z}$. Whilst there are a myriad of other schemes that make use of algebras to define a cryptosystem (see for example \cite{newhope}, \cite{kyber}, \cite{falcon}), concerns have been raised regarding the security of such schemes. Whilst the algebraic structure might allow for easier computations, the additional structure exhibited by an algebra could be exploited to allow for easier reduction of lattices based on algebras.
\\ \indent As we have already mentioned, cyclotomic polynomials exhibit many properties that are desirable in cryptography. In particular, power-of-two cyclotomic rings, that is, cyclotomic rings with conductor $N=2^n$ for some integer $n$, have found many applications. This is largely a consequence of a few properties exhibited by power-of-two cyclotomic rings: for example, $N/2$ is also a power of two, and arithmetic in the ring can be performed with ease using the $N$-dimensional FFT. However, restricting cryptosystems to only using power-of-two cyclotomic rings has its drawbacks. The most obvious of these drawbacks is the increase in dimension of the ideal lattice when moving from one power-of-two cyclotomic ring to the next, which doubles with each successive ring. The cryptosystem may require a lattice of intermediate security: for example, ideal lattices of the cyclotomic ring of conductor $1024$ have dimension $512$, but the next power-of-two cyclotomic ring has conductor $2048$. Hence, ideal lattices defined in this ring have dimension $1024$, which is a significant jump. For this reason amongst others, cryptographers have begun to move away from power-of-two cyclotomic rings to cyclotomic rings of more general conductor. However, this migration from power-of-two cyclotomics is relatively novel, and as such literature regarding the reduction of ideal lattices based on cyclotomic rings of general conductor is still lacking.
\subsection{Previous Works}
There have been a variety of studies into the shortest vector problem (SVP) in lattices generated by ideals. In 2016, Cramer, Ducas, Peikert and Regev published a paper detailing an attack on ideals generated by principal ideals of prime-power cyclotomic rings \cite{shortgen}. They presented a technique involving the use of the log-unit lattice, and showed that there is a polynomial time classical reduction from the problem of finding a short generator of a principal ideal (SPIP) to the problem of finding a generator of a principal ideal (PIP).  Moreover, in \cite{biasset}, it was shown how to solve PIP in polynomial time with a quantum computer by Biasse and Song. These results show that there is a polynomial time quantum algorithm for SPIP. We note that SPIP corresponds to finding approximate short vectors in principal ideals. Later, Cramer, Ducas, and Wesolowski gave a quantum (heuristic) reduction of the general approx-SVP algorithm to SPIP \cite{stickelb}. Therefore a quantum polynomial algorithm for approx-SVP follows from these papers.\\
\indent Simultaneously, largely inspired by Bernstein's work on subfield attacks against ideal lattices \cite{subfieldlogarithm}, Albrecht, Bai and Ducas proposed a different method to attack the NTRU cryptosystem with overstretched parameters - that is, the NTRU encryption scheme with much larger modulus \cite{subfield}. Their method entailed an attack on the NTRU cryptosystem by focusing on a sublattice, defined by a public key attained by ``norming down'' the public key of the original lattice to a subfield, and then ``lifting'' a solution on the sublattice to a solution for the original cryptosystem. Provided the solution is sufficiently good in the sublattice, it may yield a short lattice vector in the full lattice. Indeed, there are many examples of previous works which detail lattice attacks against ideal lattices. For a detailed list regarding previous research into the reduction of ideal lattices, we refer the reader to the ``previous works'' section of \cite{idealrandomprime}.
\\ \indent Recently, Pan, Xu, Wadleigh and Cheng pioneered a remarkable technique to approach the problem of SVP in prime and general ideal lattices, obtained from power-of-two cyclotomics \cite{idealrandomprime}. Their method involves manipulating the decomposition group of prime ideals in order to significantly reduce the dimension of the lattice required to solve the SVP. Their primary contributions splits into two parts. The first part considers a number field $L$ that is Galois over $\mathbb{Q}$ and shows that, given a prime ideal of its ring of integers $\mathcal{O}_L$, if the Hermite-SVP can be solved for a certain factor in a sublattice generated by a subideal, this yields a solution for the Hermite-SVP in the original ideal lattice with a larger factor, where the factor's increase depends only on the square root of the degree of $L$ over $\mathbb{Q}$ divided by the size of the decomposition group. 
\\ The second part of their paper is dedicated to ideals over the ring of integers of cyclotomic fields of conductor $N=2^{n+1}$. Under the so-called coefficient embedding, they showed that using a subgroup of the decomposition group of a prime ideal, the shortest vector in the ideal is equivalent to the shortest vector in a subideal constructed in the paper, and so solving SVP over such ideals is easy given an oracle to solve SVP in lattices of dimension equivalent to the dimension of the ideal lattice generated by the subideal. Moreover, if such a prime ideal lies above a rational prime $p$ of the form $p \equiv \pm 3 \mod 8$ then the shortest vector is of length $\sqrt{p}$, and is very easy to determine. In the final section, they showed that their method also worked for general ideals by considering the prime decomposition of an ideal.
\subsection{Our Results}
In this paper, we generalise the results of Pan et. al., both in their work on the Hermite-SVP for prime ideal lattices and solving SVP exactly for prime ideals in cyclotomic ideals. The first half of the paper is dedicated to the Hermite-SVP in ideal lattices. Whilst Pan et. al. only covered the case for lattices based on prime ideals lying above unramified primes, we extend their result to the case of general ideals whose prime ideal factors all lie over unramified primes, showing that by solving the Hermite-SVP on a subideal with some factor $\gamma$, the solution may be lifted to yield a solution for the Hermite-SVP in the original lattice with factor $\gamma^{\prime}$, where $\gamma^{\prime}/\gamma$ depends only on the factor given by Pan et. al. multiplied by a value determined by certain properties of the ideal and its decomposition group. We take this notion even more generally, and consider a module over the ring of integers of a Galois field and provide a method where a solution to the Hermite-SVP in a submodule generating a lattice of lower dimension may be lifted to a solution in the original module lattice for the Hermite-SVP with an upper bound, where the new constant is given in terms of the old factor multiplied by some factor dependent only on the ideals used to describe the module in the pseudo-basis representation.
\\ \indent The second half of the paper focuses on prime ideals of cyclotomic rings. Our work extends the results of Pan et. al. to prime ideal lattices constructed from cyclotomic rings with more general conductors, covering the cases of a general composite conductor $N=s2^{n+1}$ and $s^{\prime} p^{n+1}$ for some odd prime $p$, odd integer $s \geq 3$ and integer $s^{\prime}$, $\gcd(s^{\prime},p)=1$, which, combined with the work of Pan et. al., covers the case for any conductor $N$. In particular, our work shows that if the prime ideal in question lies above certain primes, then the dimension in which we have to solve SVP decreases significantly.
\begin{theorem}
Let $N=s2^{n+1}$, where $n$ is a positive integer and $s \geq 3$ is an odd integer. Let $\mathfrak{p}$ be a prime ideal in the ring $\mathbb{Z}[\zeta_N]$ and suppose that $\mathfrak{p}$ contains a rational prime $\rho$, where $\rho^{\phi(s)} \equiv 3 \mod 4$. Then, given an oracle that can solve SVP for $\phi(s)$-dimensional lattices, a shortest nonzero vector in $\mathfrak{p}$ can be found in time $\text{poly}(\phi(N),\log_2 \rho)$ under the canonical embedding.
\end{theorem}
\begin{theorem}
Let $N=sp^{n+1}$, where $n$ is a positive integer, $p$ is an odd prime and $s$ is a positive integer such that $\gcd(s,p)=1$. Let $\mathfrak{p}$ be a prime ideal in the ring $\mathbb{Z}[\zeta_N]$ and suppose that $\mathfrak{p}$ contains a rational prime $\rho$, where $\rho^{\phi(s)} =lp + a$ for some integers $l,a$, $\gcd(l,p)=\gcd(a,p)=1$. Then, given an oracle that can solve SVP for $(p-1)\phi(s)$-dimensional lattices, a shortest nonzero vector in $\mathfrak{p}$ can be found in time $\text{poly}(\phi(N),\log_2 \rho)$ under the canonical embedding.
\end{theorem}
\noindent As with the case for conductor $N=2^{n+1}$, we must ask whether the ``average case'' of prime ideal SVP over such cyclotomic rings is easy. This question in itself is ill-defined, and depends on how we define the distribution from which we choose the prime ideal. As we show in section 7, if we were to pick our ideal by uniformly choosing an ideal from the set of prime ideals whose rational prime lies below a certain bound, the probability of choosing an easily solvable ideal lattice is non-negligible. However, if we are to uniformly choose from the distribution of ideals of norm less than a certain bound, the probability of choosing an easily solvable ideal lattice is negligible.
\\ \indent In the last few sections, we also cover the case of general cyclotomic ideals and modules defined over a pseudo basis of ideals and vectors. For the case of general ideals, in a similar fashion to that in Pan et. al.'s work, we analyse SVP by studying the prime decomposition of ideals, and show that the shortest nonzero vector in a general ideal can be found by finding the shortest nonzero vector in a subideal of a smaller dimension. Moreover, the algorithm used to tackle SVP in such a lattice does not use the prime decomposition of the ideal, which is the most computationally complex step after SVP. In the module case, using sufficient generalisations of the decomposition group applied to modules, we show that SVP in the original module can be solved by finding the shortest nonzero vector in a submodule which has smaller dimension as a lattice after canonically embedding, and provide an algorithm by which we can solve SVP in such structures.
\\ \indent Whilst the work on ideal-SVP may initially appear to destabilise the security of cryptosystems based on cyclotomic ideals, we must point out that this work does not break Ring-LWE. Though Ideal-SVP underpins the security of Ring-LWE, our work does not directly impact the security of these schemes, since the worst-case to average-case security reduction is one-way. However, our results on module SVP may cause concern for the security of MLWE, and hence RLWE. Unlike ideal SVP and RLWE, module SVP and MLWE are known to be (polynomially) equivalent, as are MLWE and RLWE. In section 6, we offer an algorithm that can be used to find the shortest vector in a general module lattice over a cyclotomic ring by solving SVP in a submodule defined over a cyclotomic ring of lesser degree, which decreases the dimension of the lattice required to perform SVP under the canonical embedding significantly in some cases.
\subsection{Paper organisation}
The paper is organised as follows. Section 2 covers the mathematical preliminaries, including a definition of lattices and their various properties, some basic algebraic number theory and ideal lattices. The preliminary section ends with some useful lemmas regarding the factorisation of polynomials over finite fields. Section 3 covers a reduction Hermite-SVP in ideal lattices and module lattices based over a Galois extension of $\mathbb{Q}$. Section 4 presents a reduction of SVP for prime ideals of cyclotomic rings of general conductor, and show that some special cases of prime ideals are much easier to perform SVP for than others. In section 5, it is shown that our method for prime ideal lattices may be lifted to the case of general ideal lattices of cyclotomic rings. In section 6, we show that modules over cyclotomic rings may be subject to a similar reduction of SVP by studying the decomposition group of the module, and provide an algorithm which can be used to solve SVP in module lattices over cyclotomic rings. In section 7, we discuss the average-case hardness of SVP for ideal and module lattices of cyclotomic rings. 
\section{Mathematical Preliminaries}
\subsection{Lattices and the Shortest Vector Problem}
A lattice is a discrete additive subgroup of $\mathbb{R}^D$. A lattice $L$ has a basis $B=\{\mathbf{b}_1,\dots,\mathbf{b}_d\}, \mathbf{b}_i \in \mathbb{R}^D$ for some integer $d \leq D$, and every lattice point may be represented by the linear sum of basis vectors over the integers, that is,
\begin{align*}
    L=L(B)=\left\{\sum_{i=1}^dx_i \mathbf{b}_i: x_i \in \mathbb{Z}\right\}.
\end{align*}
We say that $L$ is full-rank if $d=D$. The determinant of $L$, $\det(L)$, is the square root of the volume of the fundamental parallelopiped generated by the lattice basis. If $L$ is full-rank, then $\det(L)=|\det(B)|$.
\\
In cryptography, the security of a lattice-based cryptosystem in most cases boils down to the computational hardness of the shortest vector problem (SVP). The problem can be stated as follows. Given a lattice $L$ with basis $B$, find the shortest nonzero vector in $L$ with respect to the Euclidean (or otherwise specified) norm. Most cryptosystems, however, loosen the requirement of finding the shortest nonzero vector, and require the assailant to find a nonzero vector within some range of the shortest vector. One such problem is known as the Hermite-SVP, and goes as follows.
\begin{definition}
Let $L$ be a rank $N$ lattice. The $\gamma$-Hermite-SVP is to find a nonzero lattice vector $\mathbf{v} \in L$ that satisfies
\begin{align*}
    \|\mathbf{v}\| \leq \gamma \det(L)^{1/N},
\end{align*}
for some approximation factor $\gamma \geq 1$.
\end{definition}
\noindent As opposed to the shortest lattice vector, the determinant of a lattice is well-defined and can be verified easily. Moreover, as discussed in \cite{idealrandomprime}, a solution to the Hermite-SVP can be lifted to a solution for a variety of different SVP-related problems.
\subsection{Algebraic Number Theory}
An algebraic number field $L$ is a finite extension of $\mathbb{Q}$ by some algebraic integer $\alpha$, that is, the solution to a polynomial in $\mathbb{Z}[x]$. The degree of $L$ over $\mathbb{Q}$ is equivalent to the degree of the minimal polynomial of $\alpha$ in $\mathbb{Q}(x)$. We denote by $\mathcal{O}_L$ the ring of integers of $L$, which is the maximal order of $L$. It is well known in algebraic number theory that any algebraic number field $L$ is a $\mathbb{Q}$-vector space over the power basis $\{1,\theta,\theta^2,\dots, \theta^{N-1}\}$ for some $\theta \in L$, and similarly the ring of integers $\mathcal{O}_L$ may be expressed as a $\mathbb{Z}$-module over a power basis $\{1,\theta^{\prime},\dots, {\theta^{\prime}}^{N-1}\}$ for some $\theta^{\prime} \in \mathcal{O}_K$ \cite{algnumbertheory}.
\\ \indent For a positive integer $N$, the cyclotomic polynomial $\Phi_N(x)$ is the polynomial given by
\begin{align*}
    \Phi_{N}(x)=\prod_{k=1: \gcd(k,N)=1}^N\left(x-\exp\left(\frac{2\pi i k}{N}\right)\right),
\end{align*}
or in other words, the polynomial whose roots are all the primitive $N$th roots of unity. For ease of notation, we generally let $\zeta_N=\exp\left(\frac{2\pi i}{N}\right)$ denote the $N$th root of unity. The field $L=\mathbb{Q}(\zeta_N)$ obtained by appending $\zeta_N$ to $\mathbb{Q}$ is called the cyclotomic field of conductor $N$, and such a field is of degree $\phi(N)$ over $\mathbb{Q}$, where
\begin{align*}
    \phi(N) = N\prod_{p \mid N: p\hspace{0.5mm} \text{prime}}\left(1-\frac{1}{p}\right)
\end{align*}
is Euler's totient function, which measures the number of integers less than or equal to $N$ which are coprime to $N$.
\\ The embeddings $\sigma$ of a number field $L$ are the injective homomorphisms from $L$ to $\mathbb{C}$ which fix $\mathbb{Q}$. The number of distinct embeddings is equivalent to the degree of $L$ over $\mathbb{Q}$, and an embedding $\sigma$ is said to be a real embedding if $\sigma(L) \subset \mathbb{R}$, and is said to be a complex embedding if $\sigma(L) \not\subset \mathbb{R}$. We define the \emph{canonical embedding} $\Sigma_L$ from a number field $L$ of degree $N$ to $\mathbb{C}^N$ by
\begin{align*}
    \Sigma_L: L \to \mathbb{C}^N \hspace{2mm} a \mapsto (\sigma_1(a),\sigma_2(a),\dots, \sigma_N(a)).
\end{align*}
Moreover, we respectively define the trace and norm of an element in $L$ by
\begin{align*}
    \trace_{L/\mathbb{Q}}(a) \defeq \sum_{i=1}^N\sigma_i(a), \hspace{2mm} \norm_{L/\mathbb{Q}}(a)=\prod_{i=1}^N\sigma_i(a).
\end{align*}
Defining by $\overline{a}$ the complex conjugate of an element $a \in L$, note that $\beta(x,y) \defeq \trace_{L/\mathbb{Q}}(x\overline{y})$ for all $x,y \in L$ defines a positive-definite bilinear form on the $\mathbb{Q}$-vector space generated by $L$. 
\\ \indent Another way of embedding a number field $L$ into $\mathbb{C}^N$ is using the so-called \emph{coefficient embedding}. By expressing $L$ by a $\mathbb{Q}$-vector space over a power basis $\{1,\alpha,\dots, \alpha^{N-1}\}$, we may take any element $a=\sum_{i=0}^{N-1}a_i\alpha^i$ where $a_i \in \mathbb{Q}$ and define the embedding map
\begin{align*}
    a \mapsto (a_0,a_1,\dots, a_{N-1}).
\end{align*}
It is important to note that the coefficient embedding depends on the choice of basis for $L$. Also, if $\mathcal{O}_K=\mathbb{Z}[\alpha]$ then $\mathcal{O}_K$ maps to $\mathbb{Z}^N$ under the coefficient embedding.
\subsection{Ideal Lattices}
An ideal $\mathcal{I}$ is a subring of the ring of integers $\mathcal{O}_L$. We say that an ideal $\mathfrak{p}$ is prime if, for any $ab \in \mathfrak{p}$ for $a,b \in \mathcal{O}_L$, then either $a$ or $b$ is an element of $\mathfrak{p}$. Under the canonical embedding, an ideal forms a lattice in $\mathbb{R}^N$, and we call lattices constructed in this way \emph{ideal lattices}. The volume of an ideal lattice $\mathcal{I}$ in $\mathbb{R}^N$ is $\norm_{L/\mathbb{Q}}(\mathcal{I})disc(L/\mathbb{Q})$, where $\norm_{L/\mathbb{Q}}(\mathcal{I})$ is the norm of the ideal $\mathcal{I}$ and is equivalent to the cardinality of $\mathcal{O}_L/\mathcal{I}$ (roughly speaking, the ``density'' of the ideal in $\mathcal{O}_L$), and $disc(L/\mathbb{Q})$ is the discriminant of $L$ over $\mathbb{Q}$, which is equivalent to the volume of the lattice generated after embedding $\mathcal{O}_L$ via the canonical embedding.
\\ \indent In lattice-based cryptography, the cyclotomic number field $L=\mathbb{Q}(\zeta_N)$ is frequently used to define an ideal lattice. The ring of integers of $L$ is $\mathcal{O}_L=\mathbb{Z}[\zeta_N]$ for any conductor $N$ \cite{washington}. Suppose that the cyclotomic polynomial $\Phi_N(x)$ factors in the finite field as
\begin{align*}
    \Phi_N(x)=\prod_{i=1}^g f_i(x)^e \mod p
\end{align*}
for some rational prime $p$, where $f_i(x)$ are irreducible mod $p$. Then the ideal $p\mathcal{O}_L$ factors as
\begin{align*}
    p\mathcal{O}_L=(\mathfrak{p}_1\mathfrak{p}_2\dots \mathfrak{p}_g)^e,
\end{align*}
where each $\mathfrak{p}_i=\langle p,f(\zeta_N) \rangle$ are prime ideals. We say the ideal $\mathfrak{p}_i$ \emph{lies over} $p$. If $e>1$, then $p$ is said to be \emph{ramified} in $\mathcal{O}_L$, and otherwise ($e=1$) $p$ is \emph{unramified} in $\mathcal{O}_L$. As such, we are motivated to study the factorisation of cyclotomic polynomials over finite fields in order to better study the structure of prime ideals. However, before we delve into more technical details regarding the factorisation of polynomials over finite fields, we introduce the following definition, which will be a recurring theme throughout the paper, and will be a powerful tool to help tackle SVP in ideal lattices.
\begin{definition}
Let $L/\mathbb{Q}$ be a finite Galois extension of degree $N$, and let $G$ be the Galois group of $L$ over $\mathbb{Q}$. The \emph{decomposition group} $D$ of a prime ideal $\mathfrak{p}$ is a subgroup of $G$ satisfying
\begin{align*}
    D=\{\sigma \in G: \sigma(\mathfrak{p})=\mathfrak{p}\},
\end{align*}
that is, the embeddings of $L$ that fix $\mathfrak{p}$. Then the decomposition field $K$ of $\mathfrak{p}$ is defined by
\begin{align*}
    K=\{x \in L: \forall \sigma \in D, \sigma(x)=x\},
\end{align*}
that is, the subfield of $L$ that is fixed by the decomposition group.
\end{definition}
\subsection{Factorisation of Cyclotomic Polynomials over Finite Fields}
The following lemmas will be used throughout section 4 onwards. Lemmas \ref{factornumber}-\ref{order} are standard in the study of finite fields, and we point the reader to \cite{finitefields} for more details, and also for any terminology regarding finite fields. Lemma \ref{pqdivide} is stated and proved in \cite{factorpn}.
\begin{lemma}\label{factornumber}
Let $q$ be a power of a prime and $N$ be a positive integer such that $\gcd(q,N)=1$. Then the $N$th cyclotomic polynomial $\Phi_{N}(x)$ can be factorised into $\phi(N)/m$ distinct monic irreducible polynomials of the same degree $m$ over $\mathbb{F}_q$, where $m$ is the least positive integer such that $q^m \equiv 1 \mod N$.
\end{lemma}
\begin{lemma}\label{tmonic}
Let $f_1(x),f_2(x),\dots,f_N(x)$ be distinct monic irreducible polynomials over $\mathbb{F}_q$ of degree $m$ and order $e$, and let $t \geq 2$ be an integer whose prime factors divide $e$ but not $\frac{q^m-1}{e}$. Assume that $q^m \equiv 1 \mod 4$ if $t \equiv 0 \mod 4$. Then $f_1(x^t),f_2(x^t),\dots, f_N(x^t)$ are all distinct monic irreducible polynomials of degree $mt$ and order $et$.
\end{lemma}
\begin{lemma}\label{order}
Let $f(x)$ be an irreducible polynomial over $\mathbb{F}_q$ of degree $m$ and with $f(0) \neq 0$. Then the order of $f(x)$ is equal to the order of any root of $f(x)$ in the multiplicative group $\mathbb{F}_{q^m}^*$.
\end{lemma}
\begin{lemma}\label{pqdivide}
Let $p$ be an odd prime, and $q$ be a prime power such that $q \equiv 1 \mod p$. If $m,n$ are positive integers satisfying $p^n \mid q^{p^{n-m}}-1$ and $p \nmid \frac{q^{p^{n-m}}-1}{p^n}$, then $p^{n+1} \mid q^{n+1-m}-1$ and $p \nmid \frac{q^{p^{n+1-m}}-1}{p^{n+1}}$.
\end{lemma}
\section{Solving the Hermite-SVP  for general ideal lattices in a
Galois extension}
In this section, we generalise the results of Pan et. al., specifically their contributions on the Hermite-SVP in prime ideal lattices. We consider first general ideal lattices, and then modules with a pseudo-basis of ideals and vectors with entries in the overlying number field.
\begin{defn}
Let $L/\mathbb{Q}$ be a finite Galois extension and $I\subset\mathcal{O}_L$ an ideal, expressible as $\mathcal{I} = \mathfrak{p}_1...\mathfrak{p}_g$, where each $\mathfrak{p}_i$ lies above unramified rational prime $p_i$. Let $D_\mathcal{I} = \{\sigma\in Gal(L/\mathbb{Q}):\sigma(\mathcal{I})=\mathcal{I}\}$, and $K_\mathcal{I} = L^{D_\mathcal{I}} = \{x\in L\text{ : }\sigma(x)=x, \text{ for all }\sigma\in D_\mathcal{I}\}$. These are called the \textit{decomposition group} and \textit{decomposition field} of $\mathcal{I}$, respectively.
\end{defn}
\begin{theorem}
Let ${L} / \mathbb{Q}$ be a finite Galois extension of degree $N$ and $\mathcal{I} = \mathfrak{p}_1...\mathfrak{p}_g$ an ideal of $\mathcal{O}_{{L}}$, where each $\mathfrak{p}_i$ lies over an unramified rational prime $p_i$ such that $p_i$ has $g_i$ distinct prime ideal factors in $O_{{L}}$ and has inertial degree $f_{p_i}^L$ in $\mathcal{O}_L$. If ${K}_\mathcal{I}$ is the decomposition field of $\mathcal{I}$, with $r:=[K_\mathcal{I}:\mathbb{Q}]$, and each $p_i$ has inertial degree $f_{p_i}^{K_\mathcal{I}}$ in $\mathcal{O}_{K_{\mathcal{I}}}$, then a solution to Hermite-$SVP$ with factor $\gamma$ in the sublattice $\mathfrak{c}=\mathcal{I} \cap \mathcal{O}_{{K}_\mathcal{I}}$ under the canonical embedding of ${K}_\mathcal{I}$ will also be a solution to the Hermite-SVP in $\mathcal{I}$ with factor $\gamma\frac{\sqrt{N/r}\norm_{L/\mathbb{Q}}(\mathcal{I})^{1/r-1/N}}{\norm_{{K_\mathcal{I}}/\mathbb{Q}}\left({disc}({L} / {K}_\mathcal{I})\right)^{1/2N}\left(p_1^{(f^L_{p_1}-f^{K_\mathcal{I}}_{p_1})/r}...p_g^{(f^L_{p_g}-f^{K_\mathcal{I}}_{p_g})/r}\right)}$ under the canonical embedding of ${L}.$
\end{theorem}
\begin{proof}
\noindent Consider the following diagram:
\begin{center}
\begin{tikzcd}
\mathcal{O}_L \arrow[r, hook] & L \arrow[r, "\Sigma_L"] & \mathbb{C}^N\\
\mathcal{O}_{K_I} \arrow[u, hook] \arrow[r, hook] & K_I \arrow[u, hook] \arrow[r, "\Sigma_{K_I}"] & \mathbb{C}^{[K_I:\mathbb{Q}]} \arrow[u, "\beta^\prime"]
\end{tikzcd}
\end{center}
 Here $\beta^\prime$ is chosen to make the diagram commute. Each embedding of $K_I$ extends to $N/[K_I:\mathbb{Q}]$ embeddings of $L$. Then $\beta^\prime$ simply repeats the coordinates of $\Sigma_{K_I}$ $N/r$ times, for $r = [K_I:\mathbb{Q}]$, by the definition of $K_I$. So $\|\beta^\prime(x)\| = \sqrt{N/r}\|x\|$, for any $x\in \Sigma_{K_I}(K_I)$. Set $\mathfrak{c} = I\cap\mathcal{O}_{K_I}$. Then $det(\mathfrak{c}) = \norm_{K_I/\mathbb{Q}}(\mathfrak{c})\sqrt{|disc(K_I/\mathbb{Q})|}$. So a Hermite-SVP solution $v\in\mathfrak{c}$ satisfies $\|v\|\leq\gamma(\norm_{K_I/\mathbb{Q}}(\mathfrak{c})\sqrt{|disc(K_I/\mathbb{Q})|})^{1/r}$. Also note ${disc}({L}/\mathbb{Q})={disc}({K_I}/\mathbb{Q})^{N / r} \norm_{{K_I}/\mathbb{Q}}({disc}({L} / {K}_I))$. Then 
 \begin{align*}
 \|\beta^\prime(v)\|&\leq\sqrt{N/r}\|v\|\leq\sqrt{N/r}\gamma\cdot(\norm_{K_I/\mathbb{Q}}(\mathfrak{c})\sqrt{|disc(K_I/\mathbb{Q})|})^{1/r}\\
 &= \sqrt{N/r}\gamma\cdot \norm_{K_I/\mathbb{Q}}(\mathfrak{c})^{1/r}\big(\frac{disc(L/\mathbb{Q})^{r/N}}{\norm_{{K_I}/\mathbb{Q}}({disc}({L} / {K}_I))^{r/N}}\big)^{1/2r}\\
 &= \sqrt{N/r}\gamma\cdot \norm_{K_I/\mathbb{Q}}(\mathfrak{c})^{1/r}\frac{disc(L/\mathbb{Q})^{1/2N}}{\norm_{{K_I}/\mathbb{Q}}({disc}({L} / {K}_I))^{1/2N}}\\
 &= \gamma\frac{\sqrt{N/r}}{\norm_{{K_I}/\mathbb{Q}}({disc}({L} / {K}_I))^{1/2N}}\cdot \big(\norm_{K_I/\mathbb{Q}}(\mathfrak{c})^{N/r}\sqrt{disc(L/\mathbb{Q})}\big)^{1/N}.
 \end{align*}
 The norm is multiplicative: $\norm_{K_I/\mathbb{Q}}(I) = \norm_{K_I/\mathbb{Q}}(\mathcal{P}_1)...\norm_{K_I/\mathbb{Q}}(\mathcal{P}_g)$, where $\mathcal{P}_i=\mathfrak{p}_i\cap K_I$. All the $\mathfrak{p}_i$ lie above unramified primes $p_i$, so we can write $e^{L}_{p_i}=1$. Moreover, we have $\norm_{L/\mathbb{Q}}=\norm_{K_I/\mathbb{Q}}\circ \norm_{L/K_I}$. As a result, $\norm_{L/\mathbb{Q}}(I) = p_1^{f^L_{p_1}}...p_g^{f^L_{p_g}}$. Also, $\norm_{K_I/\mathbb{Q}}(\mathfrak{c}) = p_1^{f^{K_I}_{p_1}}...p_g^{f^{K_I}_{p_g}}$. Thus $\norm_{L/\mathbb{Q}}(I) = p_1^{f^L_{p_1}}...p_g^{f^L_{p_g}} = (p_1^{f^{K_I}_{p_1}}...p_g^{f^{K_I}_{p_g}})\cdot(p_1^{f^L_{p_1}-f^{K_I}_{p_1}}...p_g^{f^L_{p_g}-f^{K_I}_{p_1}}) = \norm_{K_I/\mathbb{Q}}(\mathfrak{c})\cdot(p_1^{f^L_{p_1}-f^{K_I}_{p_1}}...p_g^{f^L_{p_g}-f^{K_I}_{p_g}})$, so we can rewrite
 \begin{align*}
 \norm_{K_I/\mathbb{Q}}(\mathfrak{c}) = \norm_{L/\mathbb{Q}}(I)/(p_1^{f^L_{p_1}-f^{K_I}_{p_1}}...p_g^{f^L_{p_g}-f^{K_I}_{p_g}}).
 \end{align*} 
 Then we have, setting $\gamma^\prime = \gamma\frac{\sqrt{N/r}}{\norm_{{K_I}/\mathbb{Q}}({disc}({L} / {K}_I))^{1/2N}}$,
\begin{align*}
\|\beta^\prime(v)\|&\leq \gamma\frac{\sqrt{N/r}}{\norm_{{K_I}/\mathbb{Q}}({disc}({L} / {K}_I))^{1/2N}}\cdot \big(\norm_{K_I/\mathbb{Q}}(\mathfrak{c})^{N/r}\sqrt{disc(L/\mathbb{Q})}\big)^{1/N}\\
&= \gamma^\prime\cdot \big(\big(\norm_{L/\mathbb{Q}}(I)/(p_1^{f^L_{p_1}-f^{K_I}_{p_1}}...p_g^{f^L_{p_g}-f^{K_I}_{p_g}})\big)^{N/r}\sqrt{disc(L/\mathbb{Q})}\big)^{1/N}\\
&= \gamma^\prime\cdot \big(\norm_{L/\mathbb{Q}}(I)^{N/r}/(p_1^{(f^L_{p_1}-f^{K_I}_{p_1})N/r}...p_g^{(f^L_{p_g}-f^{K_I}_{p_g})N/r})\sqrt{disc(L/\mathbb{Q})}\big)^{1/N}\\
&= \gamma^\prime\frac{\norm_{L/\mathbb{Q}}(I)^{1/r-1/N}}{(p_1^{(f^L_{p_1}-f^{K_I}_{p_1})1/r}...p_g^{(f^L_{p_g}-f^{K_I}_{p_g})1/r})}\cdot \big(\norm_{L/\mathbb{Q}}(I)\sqrt{disc(L/\mathbb{Q})}\big)^{1/N},
\end{align*}
as required.
 \end{proof}
 \noindent Note that when $\mathcal{I} = \mathfrak{p}$, $K_I$ is the regular decomposition field, and $f^{K_I}_p = 1$ and $f^{L}_p = N/r$, so $\norm_{L/\mathbb{Q}}(I)^{1/r-1/N}=p^{N/r(1/r-1/N)} = p^{N/r^2-1/r}$ and $p^{(f^L_{p}-f^{K_I}_{p})1/r}) = p^{N/r^2-1/r}$, and one recovers the result of \cite{idealrandomprime}.
 \subsection{Solving the Hermite-SVP for module lattices defined over a Galois extension}
The above method can be extended to the case of $\mathcal{O}_L$-modules. As before, $L$ is a field that is Galois over $\mathbb{Q}$ with ring of integers $\mathcal{O}_L$. Suppose $\mathcal{I}_1,\dots, \mathcal{I}_d \subseteq \mathcal{O}_L$ are some ideals of $\mathcal{O}_L$ and $\mathbf{b}_1,\dots, \mathbf{b}_d \in L^D$ for some integer $d \leq D$ that are linearly independent over $L$ (that is, none of the vectors can be expressed as a linear sum of the others over $L$). We define an $\mathcal{O}_L$-module $M$ with pseudo-basis $\langle \mathcal{I}_k,\mathbf{b}_k \rangle$ as the direct sum
\begin{align*}
    M=\bigoplus_{k=1}^d \mathcal{I}_k \mathbf{b}_k,
\end{align*}
which is a $\mathbb{Z}$-module of dimension $d[L:\mathbb{Q}]$. We define the volume of $M$ by
\begin{align*}
    \vol(M)=|\text{disc}(L/ \mathbb{Q})|^d\norm_{L/\mathbb{Q}}(\det(B^\dagger B))\prod_{k=1}^d \norm_{L/\mathbb{Q}}(I_k)^2,
\end{align*}
where $B$ is the matrix composed of the columns $\mathbf{b}_1,\dots, \mathbf{b}_d$ and $\dagger$ denotes the Hermitian transpose of a matrix. Then, for any $\mathbf{v}=(v_1,\dots,v_D) \in M$, the lattice $\mathcal{L}_M$ generated by $M$ under the canonical embedding is the lattice whose elements take the form $(\sigma_1(v_1),\dots,\sigma_{1}(v_D),\sigma_2(v_1),\dots,\sigma_2(v_D),\dots,\sigma_{[L:\mathbb{Q}]}(v_D))$. We say a vector $\mathbf{v}$ is a solution to the the $\gamma$-Hermite-SVP in $M$ if it satisfies $\|\mathbf{v}\| \leq \gamma \vol(M)^{1/2d[L:\mathbb{Q}]}$. By abuse of notation, we say that for any $\mathbf{v}=(v_1,\dots,v_D) \in L^D$, $\sigma(\mathbf{v})=(\sigma(v_1),\dots,\sigma(v_D))$ for all $\sigma \in \text{Gal}(L/\mathbb{Q})$ and \\$\trace_{L/K}(\mathbf{v})=\sum_{\sigma \in \text{Gal}(L/K)}\sigma(\mathbf{v})$ for any subfield $K$ of $L$. We call the module $\mathcal{I}_k \mathbf{b}_k$ a \emph{pseudo-ideal} for each $k$, and define the decomposition group $\Delta_k=\{\sigma \in \text{Gal}(L/\mathbb{Q}): \sigma(\mathcal{I}_k\mathbf{b}_k)=\mathcal{I}_k\mathbf{b}_k\}$.
\begin{lemma}\label{separate}
Let $\mathcal{I}_k\mathbf{b}_k$ be a pseudo-ideal in $L^D$, and let $\Delta_k$ denote the decomposition group of $\mathcal{I}_k\mathbf{b}_k$ and $K_k$ the decomposition field of $\Delta_k$. Then there exists an $\alpha_k \in L$ such that $\mathbf{b}_k=\alpha_k\mathbf{b}_k^{\prime}$, where $\mathbf{b}_k^{\prime} \in K_k^D$.
\end{lemma}
\begin{proof}
Let $x_k$ be some element of $\mathcal{I}_k$. By the definition of the decomposition group, for every $\sigma \in \Delta_k$ we must have $\sigma(x_k\mathbf{b}_k) \in \mathcal{I}_k \mathbf{b}_k$ and so there exists some $x_{\sigma,k} \in \mathcal{I}_k$ such that $x_{\sigma,k}\mathbf{b}_k=\sigma(x_k\mathbf{b}_k)$. Then, if we take $x_k$ to be some element of $\mathcal{O}_{K_k}$, we have
\begin{align*}
    x_k\trace_{L/{K_k}}(\mathbf{b}_k)=\trace_{L/{K_k}}(x_k\mathbf{b}_k)=\left(\sum_{\sigma \in \Delta_k}x_{k,\sigma}\right)\mathbf{b}_k \in K_k^D,
\end{align*}
and so setting $\mathbf{b}_k^{\prime}=\left(\sum_{\sigma \in \Delta_k}x_{k,\sigma}\right)\mathbf{b}_k$, $\alpha_k=\left(\sum_{\sigma \in \Delta_k}x_{k,\sigma}\right)^{-1}$, the lemma holds.
\end{proof}
Using this lemma, if we have a module equivalent to the direct sum of the pseudo-ideals $\mathcal{I}_k\mathbf{b}_k$ with decomposition groups $\Delta_k$ and decomposition fields $K_k$, then we may represent the module as
\begin{align*}
    M=\bigoplus_{k=1}^d \alpha_k\mathcal{I}_k\mathbf{b}_k^\prime,
\end{align*}
where $\mathbf{b}_k^\prime \in K_k^D$.
\begin{theorem}
Let $M$ be a module with pseudo-basis $\langle \mathcal{I}_k,\mathbf{b}_k \rangle_{k=1}^d $, and suppose that each pseudo-ideal $\mathcal{I}_k\mathbf{b}_k$ has decomposition group $\Delta_k$ and decomposition field $K_k$. Denote by $\langle \mathcal{I}_k \alpha_k, \mathbf{b}_k^{\prime} \rangle_{k=1}^d$ the pseudo-basis that also represents $M$ such that $\mathbf{b}_k^{\prime} \in K_k^D$. Let $\mathcal{J}_k=Q_k\alpha_k\mathcal{I}_k$ where $Q_k$ is a rational integer such that $\mathcal{J}_k$ is an ideal of $\mathcal{O}_L$ and $\mathcal{J}_k=\prod_{i=1}^{g_k}\mathfrak{p}_i^{(k)}$ where $\mathfrak{p}_i^{(k)}$ are prime ideals lying above an unramified rational prime $p_i^{(k)}$ with inertial degree $f_{p_i^{(k)}}^L$ in $\mathcal{O}_L$ and $f_{p_i^{(k)}}^{K_k}$ in $\mathcal{O}_{K_k}$. We let $\mathfrak{c}_k=\mathcal{J}_k \cap \mathcal{O}_{K_k}$ and let $\mathcal{M}=\bigoplus_{k=1}^d \frac{1}{Q_k}\mathfrak{c}_k\mathbf{b}_k^{\prime}$. If $K$ is the compositum of all $K_k$, $1 \leq k \leq d$, then a solution to the $\gamma$-Hermite-SVP under the canonical embedding yields a solution to the $\gamma^\prime$-Hermite-SVP in $M$, where
\begin{align*}
    \gamma^{\prime}=\gamma \frac{\sqrt{[L:K]}}{\sqrt{|\norm_{K/\mathbb{Q}}(\text{disc}(L/K))}^{\frac{1}{[L:\mathbb{Q}]}}}\prod_{k=1}^d\frac{\norm_{L/\mathbb{Q}}(\mathcal{J}_k)^{\frac{1}{d[K:\mathbb{Q}]}-\frac{1}{d[L:\mathbb{Q}]}}}{\prod_{i=1}^{g_k}{p_i^{(k)}}^{\frac{1}{d[K:\mathbb{Q}]}\left(f_{p_i^{(k)}}^L-[K:K_k]f_{p_i^{(k)}}^{K_k}\right)}}.
\end{align*}
\end{theorem}
\begin{proof}
\noindent Consider the following diagram:
\begin{center}
\begin{tikzcd}
M \arrow[r, hook] & L^D \arrow[r, "\Sigma_{L^D}"] & \mathbb{C}^{d[L:\mathbb{Q}]}\\
\mathcal{M} \arrow[u, hook] \arrow[r, hook] & K^D \arrow[u, hook] \arrow[r, "\Sigma_{K^D}"] & \mathbb{C}^{d[K:\mathbb{Q}]} \arrow[u, "\beta^\prime"]
\end{tikzcd}
\end{center}
Here, $\beta^\prime$ is chosen so that the diagram commutes. Each embedding of $K$ extends to $[L:K]$ embeddings of $L$, so $\beta^{\prime}$ repeats the coordinates of $\Sigma_{K^D}$ $[L:K]$ times by the definition of $K$ being the compositum of fixed fields, so $\|\beta^{\prime}(\mathbf{v})\|=\sqrt{[L:K]}\|\mathbf{v}\|$ for any $\mathbf{v} \in \Sigma_{K^D}(K)$. Since the norm is multiplicative, $\norm_{L/\mathbb{Q}}=\norm_{K/\mathbb{Q}} \circ \norm_{L/K}$. Now, we have
\begin{align*}
\norm_{L/\mathbb{Q}}(\mathcal{J}_k)=\prod_{i=1}^{g_k}{p_i^{(k)}}^{f_{p_i^{(k)}}^L}
\end{align*} 
and 
\begin{align*}
\norm_{K/\mathbb{Q}}(\mathfrak{c}_k)=\left(\norm_{K_k/\mathbb{Q}}(\mathfrak{c}_k)\right)^{[K:K_k]}=\prod_{i=1}^{g_k}{p_i^{(k)}}^{[K:K_k]f_{p_i^{(k)}}^{K_k}},
\end{align*}
and so we have the representation
\begin{align*}
\norm_{K/\mathbb{Q}}(\mathfrak{c}_k)=\frac{\norm_{L/\mathbb{Q}}(\mathcal{J}_k)}{\prod_{i=1}^{g_k}{p_i^{(k)}}^{f_{p_i^{(k)}}^L-[K:K_k]f_{p_i^{(k)}}^{K_k}}}.
\end{align*}
Therefore,
\begin{align*}
    &\vol(\mathcal{M})=|\text{disc}(K/\mathbb{Q})|^d\norm_{K/\mathbb{Q}}(B^{\dagger}B)\prod_{k=1}^d\norm_{K/\mathbb{Q}}(\mathfrak{c}_k)^2\prod_{k=1}^dQ_k^{-2[K:\mathbb{Q}]}
    \\&=|\text{disc}(K/\mathbb{Q})|^d\norm_{K/\mathbb{Q}}(B^{\dagger}B)\prod_{k=1}^d\frac{\norm_{L/\mathbb{Q}}(\mathcal{J}_k)^2}{Q_k^{2[K:\mathbb{Q}]}\prod_{i=1}^{g_k}{p_i^{(k)}}^{2\left(f_{p_i^{(k)}}^L-[K:K_k]f_{p_i^{(k)}}^{K_k}\right)}},
\end{align*}
where $B$ is the matrix made up composed of the vectors $\mathbf{b}_1^{\prime},\dots, \mathbf{b}_d^{\prime}$, and using the fact that $\text{disc}(L/\mathbb{Q})=\text{disc}(K/\mathbb{Q})^{[L:K]}\norm_{K/\mathbb{Q}}(\text{disc}(L/K))$,
\begin{align*}
    &\vol(M)=|\text{disc}(L/\mathbb{Q})|^d\norm_{L/\mathbb{Q}}(B^{\dagger}B)\prod_{k=1}^d\norm_{L/\mathbb{Q}}(\mathcal{J}_k)^2\prod_{k=1}^dQ_k^{-2[L:\mathbb{Q}]}
    \\&=|\text{disc}(K/\mathbb{Q})|^{d[L:K]}\norm_{K/\mathbb{Q}}(\text{disc}(L/K))^d\\&\cdot \norm_{K/\mathbb{Q}}(B^{\dagger}B)^{[L:K]}\prod_{k=1}^d\norm_{L/\mathbb{Q}}(\mathcal{J}_k)^2\prod_{k=1}^dQ_k^{-2[L:\mathbb{Q}]}.
\end{align*}
Hence, if we have a solution $\mathbf{v} \in \mathcal{M}$ for the the Hermite-SVP with factor $\gamma$, then
\begin{align*}
    &\|\beta(\mathbf{v})\|^2=[L:K]\|\mathbf{v}\|^2\leq \gamma^2 [L:K] \vol(\mathcal{M})^{1/d[K:\mathbb{Q}]}
    \\&= \gamma^2 [L:K]|\text{disc}(K/\mathbb{Q})|^{1/[K:\mathbb{Q}]}\norm_{K/\mathbb{Q}}(B^{\dagger}B)^{1/d[K:\mathbb{Q}]}\\ &\cdot\prod_{k=1}^d\frac{\norm_{L/\mathbb{Q}}(\mathcal{J}_k)^{2/d[K:\mathbb{Q}]}}{Q_k^{2/d}\prod_{i=1}^{g_k}{p_i^{(k)}}^{\frac{2}{d[K:\mathbb{Q}]}\left(f_{p_i^{(k)}}^L-[K:K_k]f_{p_i^{(k)}}^{K_k}\right)}}
    \\&=\frac{\gamma^2[L:K]}{|\norm_{K/\mathbb{Q}}(\text{disc}(L/K))|^{1/[L:\mathbb{Q}]}}\vol(M)^{1/d[L:\mathbb{Q}]}\\& \cdot\prod_{k=1}^d\frac{\norm_{L/\mathbb{Q}}(\mathcal{J}_k)^{\frac{2}{d}\left(\frac{1}{[K:\mathbb{Q}]}-\frac{1}{[L:\mathbb{Q}]}\right)}}{\prod_{i=1}^{g_k}{p_i^{(k)}}^{\frac{2}{d[K:\mathbb{Q}]}\left(f_{p_i^{(k)}}^L-[K:K_k]f_{p_i^{(k)}}^{K_k}\right)}},
\end{align*}
as required.
\end{proof}
Though the factor we have obtained for general ideal lattices and module lattices may seem somewhat convoluted and tricky to interpret, we may make two remarks. The first, if $p_i^{(k)}\mathcal{O}_{K_k}=\prod_{j=1}^{t_{i,k}}\mathfrak{p}_{i,j}^{(k)}$ where $\mathfrak{p}_{i,j}^{(k)}$ are prime ideals of $\mathcal{O}_{K_k}$ and each $\mathfrak{p}_{i,j}^{(k)}$ is inert in $L$ for all $i,j,k$, then $\prod_{k=1}^d\frac{\norm_{L/\mathbb{Q}}(\mathcal{J}_k)^{\frac{1}{d}\left(\frac{1}{[K:\mathbb{Q}]}-\frac{1}{[L:\mathbb{Q}]}\right)}}{\prod_{i=1}^{g_k}{p_i^{(k)}}^{\frac{1}{d[K:\mathbb{Q}]}\left(f_{p_i^{(k)}}^L-[K:K_k]f_{p_i^{(k)}}^{K_k}\right)}}=1$ (set $d=1$ for the ideal lattice case). Secondly, we may attain an upper bound that is easier to comprehend:
\begin{align*}
    &\gamma \frac{\sqrt{[L:K]}}{\sqrt{|\norm_{K/\mathbb{Q}}(\text{disc}(L/K))}^{\frac{1}{[L:\mathbb{Q}]}}}\prod_{k=1}^d\frac{\norm_{L/\mathbb{Q}}(\mathcal{J}_k)^{\frac{1}{d[K:\mathbb{Q}]}-\frac{1}{d[L:\mathbb{Q}]}}}{\prod_{i=1}^{g_k}{p_i^{(k)}}^{\frac{1}{d[K:\mathbb{Q}]}\left(f_{p_i^{(k)}}^L-[K:K_k]f_{p_i^{(k)}}^{K_k}\right)}}\\
    &\leq \gamma \sqrt{[L:K]} \prod_{k=1}^d \norm_{L/\mathbb{Q}}(\mathcal{J}_k)^{\frac{1}{d[K:\mathbb{Q}]}-\frac{1}{d[L:\mathbb{Q}]}},
\end{align*}
and this value can be determined without having to know the prime decomposition of the ideals.
\section{Prime Ideals of Cyclotomic Fields}
\subsection{The Cyclotomic Field $L=\mathbb{Q}(\zeta_{s2^{n+1}})$}
We let $s$ be some positive odd integer, $s \geq 3$. The following is Theorem 2.2 from \cite{wangwang}. 
\begin{theorem}\label{s2factorisation}
Let $q$ be an odd prime power, and let $s \geq 3$ be any odd number such that $\gcd(q,s)=1$, and let $q^{\phi(s)}=m2^A+1$ for some odd $m$, $A \geq 1$. Then, for any $A -1\leq n$ and for any irreducible factor $f(x)$ of $\Phi_{s2^A}(x)$ over $\mathbb{F}_q$, then $f(x^{2^{n-A+1}})$ is also irreducible over $\mathbb{F}_q$. Moreover, all irreducible factors of $\Phi_{s2^{n+1}}(x)$ are obtained in this way.
\end{theorem}
\begin{theorem}\label{s2svp}
For any prime ideal $\mathfrak{p}=\langle \rho,f(\zeta_{s2^{n+1}}) \rangle$ of $\mathcal{O}_L$ for some rational prime $\rho$, $\gcd(\rho,s)=\gcd(\rho,2)=1$ and irreducible polynomial $f(x)$ of $\Phi_{s2^{n+1}}$ in $\mathbb{F}_\rho[x]$, write $\rho^{\phi(s)}=m2^A+1$ where $m$ is an odd integer and $A \geq 1$, and let $r=\min\{A-1,n\}$. Then, given an oracle that can solve SVP for $\phi(s2^{r+1})$-dimensional lattices, a shortest nonzero vector in $\mathfrak{p}$ can be found in \\$\text{poly}(\phi(s2^{n+1}),\log_2\rho)$ time with the canonical embedding.
\end{theorem}
\begin{proof}
We assume that $n \geq A$ otherwise the theorem is vacuously true, so $r=A-1$. Let
\begin{align*}
    G=\{\sigma_i: \gcd(i,2)=\gcd(i,s)=1\}
\end{align*}
denote the Galois group of $L$ over $\mathbb{Q}$, where
\begin{align*}
    &\sigma_i: \mathbb{Q}(\zeta_{s2^{n+1}}) \to  \mathbb{Q}(\zeta_{s2^{n+1}}),\\
    &\sigma_i(\zeta_{s2^{n+1}}^k)=\zeta_{s2^{n+1}}^{ki}.
\end{align*}
By Theorem \ref{s2factorisation}, for any factor $f(x)$ of $\Phi_{s2^{n+1}}(x)$ that is irreducible in $\mathbb{F}_\rho[x]$, there exists a polynomial $g(x)$ that is a factor of $\Phi_{s2^{r+1}}(x)$ that is irreducible over $\mathbb{F}_{\rho}[x]$ such that $f(x)=g(x^{2^{n-r}})$. Then the prime ideal lattice $\mathfrak{p}$ can be represented by
\begin{align*}
    \langle \rho,f(\zeta_{s2^{n+1}}) \rangle = \langle \rho,g(\zeta_{s2^{r+1}}) \rangle.
\end{align*}
For any $1 \leq k \leq 2^{n-r}-1$, the map $\sigma_{ks2^{r+1}+1}$ fixes $\zeta_{s2^{n+1}}^{l2^{n-r}}$ for any integer $0 \leq l <s2^{r+1}$. Moreover, since $\gcd(ks2^{r+1}+1,2)=\gcd(ks2^{r+1},s)=1$, each subset $H_k$ of $G$ generated $\sigma_{ks2^{r+1}+1}$ forms a cyclic group, and so the set $H=H_1 \times H_2 \times \dots \times H_{2^{n-r}-1}$ forms a subgroup of the decomposition group of $\mathfrak{p}$, since both $\rho$ and $f(\zeta_{s2^{n+1}})=g(\zeta_{2^{n+1}}^{2^{n-r}})$ are fixed by each $\sigma_i \in H$. $K=\mathbb{Q}(\zeta_{s2^{n+1}}^{2^{n-r}})$ must be the fixed field of the group $H$, as for all $i \in \left(\mathbb{Z}/s2^{n+1}\mathbb{Z}\right)^{\times}$,
\begin{align*}
    \sigma_i(\zeta_{s2^{n+1}}^{2^{n-r}})=\zeta_{s2^{n+1}}^{2^{n-r}} \iff i \equiv 1 \mod s2^{r+1}.
\end{align*}
Note that $\mathcal{O}_K$ has the $\mathbb{Z}$-basis $\{1,\zeta_{s2^{n+1}}^{2^{n-r}},\zeta_{s2^{n+1}}^{2(2^{n-r})}, \dots, \zeta_{s2^{n+1}}^{(\phi(s2^{r+1})-1)(2^{n-r})}\}$. Letting $\mathfrak{c}=\mathfrak{p} \cap \mathcal{O}_K$, we claim that
\begin{align*}
    \mathfrak{p}=\bigoplus_{k=0}^{2^{n-r}}\zeta_{s2^{n+1}}^k\mathfrak{c}.
\end{align*}
For any $a \in \mathfrak{p}$, there exist integers $z_i,w_i$ such that
\begin{align*}
    a&=\sum_{i=0}^{\phi(s2^{n+1})-1}z_i\zeta_{s2^{n+1}}^if(\zeta_{s2^{n+1}})+\sum_{i=0}^{\phi(s2^{n+1})-1}w_i\zeta_{s2^{n+1}}\rho
    \\&=\sum_{k=0}^{2^{n-r}-1}\zeta_{s2^{n+1}}^k\sum_{j=0}^{\phi(s2^{r+1})-1}\left(z_{k+j2^{n-r}}\zeta_{s2^{n+1}}^{j2^{n-r}}f(\zeta_{s2^{n-r}})+w_{k+j2^{n-r}}\zeta_{s2^{n+1}}^{j2^{n-r}}\rho\right)
    \\&=\sum_{k=0}^{2^{n-r}-1}\zeta_{s2^{n+1}}^k\Bigg(\left(\sum_{j=0}^{\phi(s2^{r+1})-1}z_{k+j2^{n-r}}\zeta_{s2^{n+1}}^{j2^{n-r}}\right)f(\zeta_{s2^{n+1}})\\&+\left(\sum_{j=0}^{\phi(s2^{r+1}-1}w_{k+j2^{n-r}}\zeta_{s2^{n+1}}^{j2^{n-r}}\right)\rho\Bigg),
\end{align*}
which proves our claim. Now, for any $x_k \in \mathfrak{c}, 0 \leq k \leq 2^{n-r}-1$, let $x=\sum_{k=0}^{2^{n-r}-1}x_k \zeta_{s2^{n+1}}^k \in \mathfrak{p}$. Then the quadratic form induced by the ideal lattice $\mathfrak{p}$ is given by
\begin{align*}
    &\trace_{L/\mathbb{Q}}(x\overline{x})=\trace_{L/\mathbb{Q}}\left(\sum_{k,l=0}^{2^{n-r}-1}x_k\overline{x_l}\zeta_{s2^{n+1}}^{k-l}\right)\\&=\sum_{i=0: \gcd(i,s)=\gcd(i,2)=1}^{s2^{n+1}-1}\sum_{k,l=0}^{2^{n-r}-1}\sigma_i\left(x_k\overline{x_l}\zeta_{s2^{n+1}}^{k-l}\right)
    \\&=\sum_{i=0: \gcd(i,s)=\gcd(i,2)=1}^{s2^{r+1}-1}\sum_{j=0}^{2^{n-r}-1}\sum_{k,l=0}^{2^{n-r}-1}\sigma_{i+js2^{r+1}}(x_k\overline{x_l})\zeta_{s2^{n+1}}^{(i+js2^{r+1})(k-l)}
    \\&=\sum_{i=0: \gcd(i,s)=\gcd(i,2)=1}^{s2^{r+1}-1}\sum_{j=0}^{2^{n-r}-1}\sum_{k,l=0}^{2^{n-r}-1}\sigma_i(x_k\overline{x_l})\zeta_{s2^{n+1}}^{(i+js2^{r+1})(k-l)},
\end{align*}
but note that we have
\begin{align*}
    \sum_{j=0}^{2^{n-r}-1}\zeta_{s2^{n+1}}^{(i+js2^{r+1})(k-l)}=\sum_{j=0}^{2^{n-r}-1}\zeta_{s2^{n+1}}^{i(k-l)}\zeta_{2^{n-r}}^{j(k-l)}=
    \begin{cases}
    2^{n-r} \hspace{2mm} &\text{if} \hspace{2mm} k=l,
    \\ 0 \hspace{2mm} &\text{otherwise.}
    \end{cases}
\end{align*}
Hence
\begin{align*}
    \trace_{L/\mathbb{Q}}(x\overline{x})&=2^{n-r}\sum_{i=0: \gcd(i,s)=\gcd(i,2)=1}^{s2^{r+1}-1}\sum_{k=0}^{2^{n-r}-1}\sigma_i(x_k\overline{x_k})\\&=2^{n-r}\sum_{k=0}^{2^{n-r}-1}\trace_{K/\mathbb{Q}}(x_k\overline{x_k}),
\end{align*}
and so $\lambda_1(\mathfrak{p})=\lambda_1(\mathfrak{c})$, as required. The algorithm below summarises how to find the shortest nonzero vector in a prime ideal lattice $\mathfrak{p}$. The most time-consuming step in the algorithm below is Step 2, and all other steps may be performed in $\text{poly}(\phi(s2^{n+1}),\log_2 \rho)$ time.
\end{proof}
\begin{algorithm}
\SetKwInOut{Input}{input}\SetKwInOut{Output}{output} \Input{A prime ideal $\mathfrak{p}=\langle \rho, f(\zeta_{s2^{n+1}}) \rangle$ in $\mathbb{Z}[\zeta_{s2^{n+1}}]$, where $\rho$ is odd and $\gcd(\rho,s)=1$.} \Output{A shortest vector in the corresponding prime ideal lattice.} \BlankLine 
\nl Compute the ideal $\mathfrak{c}$ generated by $\rho$ and $f(\zeta_{s2^{n+1}})$ in $\mathcal{O}_K$ where $K=\mathbb{Q}(\zeta_{s2^{n+1}}^{2^{n-r}})$.
\\
\nl Find a shortest vector $v$ in the $\phi(s2^{r+1})$-dimensional lattice $\mathfrak{c}$.
\\
\nl Output v.
\caption{SVP algorithm for prime ideal lattices of $\mathbb{Z}[\zeta_{s2^{n+1}}]$}
\end{algorithm}
\subsection{The Cyclotomic Field $L=\mathbb{Q}(\zeta_{sp^{n+1}})$}
The following theorem is a generalisation of Theorem 2 in \cite{factorpn}.
\begin{theorem}\label{spfactorisation}
Let $s,p,q$ be positive integers such that $p$ is an odd prime, $q$ is a prime power and $\gcd(s,p)=\gcd(q,p)=1$. Suppose $q^{\phi(s)} \equiv a \mod p$, for some integer $\gcd(a,p)=1$ and set $q^{\phi(s)}=mp^A+a$ for some integer $m$ such that $\gcd(m,p)=1$ and some integer $A \geq 0$. Then for any $n \geq A-1$ and any irreducible factor $f(x)$ of the cyclotomic polynomial $\Phi_{p^As}(x)$ over $\mathbb{F}_q$, $f(x^{n-A+1})$ is also irreducible. Moreover, all the irreducible factors of $\Phi_{p^{n+1}s}(x)$ are obtained in this way.
\end{theorem}
\begin{proof}
See Appendix A.
\end{proof}
\begin{theorem}\label{spsvp}
For any prime ideal $\mathfrak{p}=\langle \rho,f(\zeta_{sp^{n+1}}) \rangle$ where $\rho$ is a positive rational prime, $\gcd(\rho,s)=\gcd(\rho,p)=1$ and irreducible polynomial $f(x)$ of the cyclotomic polynomial $\Phi_{sp^{n+1}}(x)$ in $\mathbb{F}_\rho[x]$, assume that $\rho^{\phi(s)} \equiv a \mod p$ for some $\gcd(a,p)=1$ and set $\rho^{\phi(s)}=mp^A + a$ for some positive integer $m$ such that $\gcd(m,p)=1$ and some integer $A \geq 1$ and let $r=\min\{A-1,n\}$. Then, given an oracle that can solve SVP for $\phi(sp^{r+1})$-dimensional lattices, a shortest nonzero vector in $\mathfrak{p}$ can be found in $\text{poly}(\phi(sp^{n+1}),\log_2 \rho)$ time with the canonical embedding.
\end{theorem}
\begin{proof}
We assume that $n \geq A$ otherwise the theorem is vacuously true. Let
\begin{align*}
    G=\{\sigma_i, 1 \leq i \leq sp^{n+1}-1: \gcd(i,p)=\gcd(i,s)=1\}
\end{align*}
denote the Galois group of $L$ over $\mathbb{Q}$, where
\begin{align*}
    &\sigma_i: \mathbb{Q}(\zeta_{sp^{n+1}}) \to \mathbb{Q}(\zeta_{sp^{n+1}}),\\
    &\sigma_i(\zeta_{sp^{n+1}}^k)=\zeta_{sp^{n+1}}^{ki}.
\end{align*}
By Theorem \ref{spfactorisation}, for any factor $f(x)$ of $\Phi_{sp^{n+1}}(x)$ that is irreducible in $\mathbb{F}_\rho[x]$, there exists a polynomial $g(x)$ that is a factor of $\Phi_{sp^{r+1}}(x)$ that is irreducible over $\mathbb{F}_\rho[x]$ such that $f(x)=g\left(x^{p^{n-r}}\right)$. Then the prime ideal $\mathfrak{p}$ can be represented by
\begin{align*}
    \langle \rho, f(\zeta_{sp^{n+1}}) \rangle = \langle \rho, g(\zeta_{sp^{r+1}}) \rangle.
\end{align*}
For any $1 \leq k \leq p^{n-r}-1$, the map $\sigma_{ksp^{r+1}+1}$ fixes $\zeta_{sp^{n+1}}^{lp^{n-r}}$ for any integer $0 \leq l <sp^{r+1}$. Moreover, since $\gcd(ksp^{r+1}+1,p)=\gcd(ksp^{r+1},s)=1$, each subset $H_k$ of $G$ generated by $\sigma_{ksp^{r+1}+1}$ forms a cyclic group, and so the set $H=H_1 \times H_2 \times \dots \times H_{p^{n-r}-1}$  forms a subgroup of the decompositiong group of $\mathfrak{p}$, since both $\rho$ and $f(\zeta_{sp^{n+1}})=g\left(\zeta_{sp^{n+1}}^{p^{n-r}}\right)$ are fixed by each $\sigma_i \in H$. Then $K=\mathbb{Q}(\zeta_{sp^{n+1}}^{p^{n-r}})$ must be the fixed field of the group $H$, as for all $i \in \left(\mathbb{Z}/sp^{n+1}\mathbb{Z}\right)^\times$,
\begin{align*}
    \sigma_i(\zeta_{sp^{n+1}}^{p^{n-r}})=\zeta_{sp^{n+1}}^{p^{n-r}} \iff i \equiv 1 \mod sp^{r+1}.
\end{align*}
Note that $\mathcal{O}_K$ has the $\mathbb{Z}$-basis $\{1,\zeta_{sp^{n+1}}^{p^{n-r}},\zeta_{sp^{n+1}}^{2p^{n-r}},\dots, \zeta_{sp^{n+1}}^{(\phi(sp^{r+1})-1)p^{n-r}}\}$. Letting $\mathfrak{c}=\mathfrak{p} \cap \mathcal{O}_K$, we claim that
\begin{align*}
    \mathfrak{p}=\bigoplus_{k=0}^{p^{n-r}-1}\zeta_{sp^{n+1}}^k \mathfrak{c}.
\end{align*}
For any $a \in \mathfrak{p}$, there exist integers $z_i,w_i$ such that
\begin{align*}
    a&=\sum_{i=0}^{\phi(sp^{n+1})-1}z_i\zeta_{sp^{n+1}}^if(\zeta_{sp^{n+1}})+\sum_{i=0}^{\phi(sp^{n+1})-1}w_i\zeta_{sp^{n+1}}\rho
    \\&=\sum_{k=0}^{p^{n-r}-1}\zeta_{sp^{n+1}}^k\sum_{j=0}^{\phi(sp^{r+1})-1}\left(z_{k+jp^{n-r}}\zeta_{sp^{n+1}}^{jp^{n-r}}f(\zeta_{sp^{n-r}})+w_{k+jp^{n-r}}\zeta_{sp^{n+1}}^{jp^{n-r}}\rho\right)
    \\&=\sum_{k=0}^{p^{n-r}-1}\zeta_{sp^{n+1}}^k\Bigg(\left(\sum_{j=0}^{\phi(sp^{r+1})-1}z_{k+jp^{n-r}}\zeta_{sp^{n+1}}^{jp^{n-r}}\right)f(\zeta_{sp^{n+1}})\\&+\left(\sum_{j=0}^{\phi(sp^{r+1})-1}w_{k+jp^{n-r}}\zeta_{sp^{n+1}}^{jp^{n-r}}\right)\rho\Bigg),
\end{align*}
which proves our claim. Now, for any $x_k \in \mathfrak{c}, 0 \leq k \leq p^{n-r}-1$, let $x=\sum_{k=0}^{p^{n-r}-1}x_k\zeta_{sp^{n+1}}^k \in \mathfrak{p}$. Then the quadratic form induced by the ideal lattice $\mathfrak{p}$ is given by
\begin{align*}
    &\trace_{L/\mathbb{Q}}(x\overline{x})=\trace_{L/\mathbb{Q}}\left(\sum_{k,l=0}^{p^{n-r}-1}x_k\overline{x_l}\zeta_{sp^{n+1}}^{k-l}\right)\\&=\sum_{i=0: \gcd(i,s)=\gcd(i,p)=1}^{sp^{n+1}-1}\sum_{k,l=0}^{p^{n-r}-1}\sigma_i(x_k\overline{x_l}\zeta_{sp^{n+1}}^{k-l})
    \\&=\sum_{i=0: \gcd(i,s)=\gcd(i,p)=1}^{sp^{r+1}-1}\sum_{j=0}^{p^{n-r}-1}\sum_{k,l=0}^{p^{n-r}-1}\sigma_{i+jsp^{r+1}}(x_k\overline{x_l})\zeta_{sp^{n+1}}^{(i+jsp^{r+1})(k-l)}
    \\&=\sum_{i=0: \gcd(i,s)=\gcd(i,p)=1}^{sp^{r+1}-1}\sum_{j=0}^{p^{n-r}-1}\sum_{k,l=0}^{p^{n-r}-1}\sigma_i(x_k\overline{x_l})\zeta_{sp^{n+1}}^{(i+jsp^{r+1})(k-l)},
\end{align*}
but note that we have
\begin{align*}
    \sum_{j=0}^{p^{n-r}-1}\zeta_{sp^{n+1}}^{(i+jsp^{r+1})(k-l)}=\sum_{j=0}^{p^{n-r}-1}\zeta_{sp^{n+1}}^{i(k-l)}\zeta_{p^{n-r}}^{j(k-l)}=
    \begin{cases}
    p^{n-r} \hspace{2mm} &\text{if} \hspace{1mm} k=l,\\
    0 \hspace{2mm} &\text{otherwise.}
    \end{cases}
\end{align*}
Hence
\begin{align*}
    \trace_{L/\mathbb{Q}}(x\overline{x})&=p^{n-r}\sum_{i=0: \gcd(i,s)=\gcd(i,p)=1}^{sp^{r+1}-1}\sum_{k=0}^{p^{n-r}-1}\sigma_i(x_k\overline{x_k})\\&=p^{n-r}\sum_{k=0}^{p^{n-r}-1}\trace_{K/\mathbb{Q}}(x_k\overline{x_k}),
\end{align*}
and so $\lambda_1(\mathfrak{p})=\lambda_1(\mathfrak{c})$, as required. The algorithm below summarises how to find the shortest nonzero vector in a prime ideal lattice $\mathfrak{p}$. The most time-consuming step in the algorithm below is Step 2, and all other steps may be performed in $\text{poly}(\phi(sp^{n+1}),\log_2 \rho)$ time.
\end{proof}
\begin{algorithm}
\SetKwInOut{Input}{input}\SetKwInOut{Output}{output} \Input{A prime ideal $\mathfrak{p}=\langle \rho, f(\zeta_{sp^{n+1}}) \rangle$ in $\mathbb{Z}[\zeta_{sp^{n+1}}]$, where $\gcd(p,\rho)=\gcd(\rho,s)=1$ and $\rho^{\phi(s)} \equiv \pm 1 \mod p$.} \Output{A shortest vector in the corresponding prime ideal lattice.} \BlankLine 
\nl Compute the ideal $\mathfrak{c}$ generated by $\rho$ and $f(\zeta_{sp^{n+1}})$ in $\mathcal{O}_K$ where $K=\mathbb{Q}(\zeta_{sp^{n+1}}^{p^{n-r}})$.
\\
\nl Find a shortest vector $v$ in the $\phi(sp^{r+1})$-dimensional lattice $\mathfrak{c}$.
\\
\nl Output v.
\caption{SVP algorithm for prime ideal lattices of $\mathbb{Z}[\zeta_{sp^{n+1}}]$}
\end{algorithm}
\subsection{Some Special Prime Ideals of Cyclotomic Rings}
\begin{theorem}
Let $L=\mathbb{Q}(\zeta_{s2^{n+1}})$ be a cyclotomic field for some positive odd integer $s \geq 3$ and some integer $n \geq 0$. Let $\mathfrak{p}$ denote a prime ideal lying over a positive rational odd prime $\rho$ such that $\rho^{\phi(s)} \equiv 3 \mod 4$. Then, given an oracle that can solve SVP for $\phi(s)$-dimensional lattices, a shortest nonzero vector in $\mathfrak{p}$ can be found in $\text{poly}(\phi(s2^{n+1}),\log_2\rho)$ time with the canonical embedding.
\end{theorem}
\begin{proof}
For some integer $N>1$, we must have $\rho^{\phi(s)} \equiv 2l+1 \mod 2^N$, and so for some integer $k$, we have $\rho^{\phi(s)} = 1+2l+2^Nk=1+2(2^{N-1}k+l)$. Since $2^{N-1}k+l$ is an odd integer and $N$ is taken totally arbitrarily, the claim holds by Theorem \ref{s2svp}.
\end{proof}
\begin{theorem}
Let $L=\mathbb{Q}(\zeta_{sp^{n+1}})$ be a cyclotomic field for some positive integer $s$ such that $\gcd(s,p)=1$, an odd prime $p$ and some integer $n \geq 0$. Let $\mathfrak{p}$ denote a prime ideal lying over a positive rational odd prime $\rho$ such that $\rho^{\phi(s)} =lp +a $ for some integers $l,a$, $\gcd(l,p)=\gcd(a,p)=1$. Then, given an oracle that can solve SVP for $(p-1)\phi(s)$-dimensional lattices, a shortest nonzero vector in $\mathfrak{p}$ can be found in $\text{poly}(\phi(sp^{n+1}),\log_2 \rho)$ time with the canonical embedding.
\end{theorem}
\begin{proof}
For some integer $N>1$, we must have $\rho^{\phi(s)} \equiv lp +a \mod p^N$, and so for some integer $k$, we have $\rho^{\phi(s)}=m a+pl+p^{N}k= a+p(p^{N-1}k+l)$. Since $\gcd(p^{N-1}k+l,p)=1$ and $N$ is taken totally arbitrarily, the claim holds by Theorem \ref{spsvp}.
\end{proof}
\section{General Ideals of Cyclotomic Rings}
\subsection{The Cyclotomic Field $L=\mathbb{Q}(\zeta_{s2^{n+1}})$}
As before, we set $s$ to be some odd integer greater than or equal to $3$.
\begin{theorem}
Let $\mathcal{I}$ be a nonzero ideal of $\mathbb{Z}[\zeta_{s2^{n+1}}]$ with prime factorisation
\begin{align*}
    \mathcal{I}=\mathfrak{p}_1\mathfrak{p}_2\dots \mathfrak{p}_t,
\end{align*}
where $\mathfrak{p}_i=(f_i(\zeta_{s2^{n+1}}),\rho_i)$ for rational primes $\rho_i$ are (not necessarily distinct) prime ideals. If $\rho_i$ is odd, write $\rho_i^{\phi(s)}=m_i2^{A_i}+1$, for some integer $m_i$ such that $\gcd(m_i,2)=1$ and let $r=\max\{r_i\}$, where
\begin{align*}
    r_i=
    \begin{cases}
    \min\{A_i-1,n\}, \hspace{0.5mm} &\text{if} \hspace{1mm} \rho_i \equiv 1 \mod 2,\\
    n \hspace{0.5mm} &\text{if} \rho_i=2.
    \end{cases}
\end{align*}
Then SVP in the lattice generated by $\mathcal{I}$ can be solved via solving SVP in a $\phi(s2^{r+1})$-dimensional lattice.
\end{theorem}
\begin{proof}
If $r=n$ the theorem vacuously holds, so we assume otherwise. We may assume WLOG that $r=r_1$. Following the notation of Theorem \ref{s2svp}, we denote by
\begin{align*}
    G=\{\sigma_i: 1 \leq i \leq s2^{n+1}-1, \gcd(i,2)=\gcd(i,s)=1\}
\end{align*}
the Galois group of $L$, and consider the subgroup $H=H_1 \times H_2 \times \dots \times H_{2^{n-r-1}}$, where $H_k$ is the cyclic group generated by $\langle \sigma_{ks2^{r+1}+1}\rangle$, which is a subgroup of the decomposition group of every $\mathfrak{p}_i$, since $\sigma_{ks2^{r+1}+1}(\rho_i)=\rho_i$, $\sigma_{ks2^{r+1}+1}(f_i(\zeta_{s2^{n+1}}))=\sigma_{ks2^{r+1}+1}(g_i(\zeta_{s2^{r+1}}))=g_i(\zeta_{s2^{r+1}})=f_i(\zeta_{s2^{n+1}})$, where $g_i(x)$ is an irreducible factor of $\Phi_{s2^{A_i}}(x)$. As shown in Theorem \ref{s2svp}, the fixed field of $H$ is $K=\mathbb{Q}(\zeta_{s2^{n+1}}^{2^{n-r}})$, which has the ring of integers $\mathcal{O}_K=\mathbb{Z}[\zeta_{s2^{n+1}}^{2^{n-r}}]$. Let $\mathfrak{c}= \mathcal{I} \cap \mathcal{O}_K$. We claim that for any $a \in \mathcal{I}$, there exist $a^{(k)} \in \mathfrak{c}$ for $0 \leq k \leq 2^{n-r}-1$ such that
\begin{align*}
    a=\sum_{k=0}^{2^{n-r}-1}\zeta_{s2^{n+1}}^ka^{(k)}.
\end{align*}
We prove the claim via induction. When $t=1$, the claim holds by Theorem \ref{s2svp}, so we assume the claim holds for $t-1$. Letting $\overline{\mathcal{I}}=\mathfrak{p}_1\mathfrak{p}_2\dots \mathfrak{p}_{t-1}$, we have $\mathcal{I}=\mathfrak{p}_t\mathcal{I}$. It suffices to show that for any $xy$, $x \in \overline{\mathcal{I}}, y \in \mathfrak{p}_t$, there exist $b^{(k)} \in \mathcal{I} \cap \mathcal{O}_K$ for $0 \leq k \leq 2^{n-r}-1$ such that $xy=\sum_{k=0}^{2^{n-r}-1}\zeta_{s2^{n+1}}^kb^{(k)}$. By the induction assumption, there exist $x^{(i)} \in \overline{\mathcal{I}} \cap \mathcal{O}_K$, $y^{(j)} \in \mathfrak{p}_t \cap \mathcal{O}_K$, $0 \leq i,j \leq 2^{n-r}-1$ such that $x=\sum_{i=0}^{2^{n-r}-1}\zeta_{s2^{n+1}}^ix^{(i)}$ and $y=\sum_{j=0}^{2^{n-r}-1}\zeta_{s2^{n+1}}^jy^{(j)}$. Hence, we have
\begin{align*}
    xy&=\sum_{i,j=0}^{2^{n-r}-1}\zeta_{s2^{n+1}}^{i+j}x^{(i)}y^{(i)}
    \\&=\sum_{k=0}^{2^{n-r}-1}\zeta_{s2^{n+1}}^k\sum_{i+j=k}x^{(i)}y^{(j)}+\sum_{k=2^{n-r}}^{2\cdot 2^{n-r}-2}\zeta_{s2^{n+1}}^k
    \\&=\sum_{k=0}^{2^{n-r}-1}\zeta_{s2^{n+1}}^k\sum_{i+j=k}x^{(i)}y^{(j)}+\sum_{k=0}^{2^{n-r}-2}\zeta_{s2^{n+1}}^k\sum_{i+j=k+2^{n-r}}\zeta_{s2^{n+1}}^{2^{n-r}}x^{(i)}y^{(j)}
    \\&=\sum_{k=0}^{2^{n-r}-2}\zeta_{s2^{n+1}}^k\left(\sum_{i+j=k}x^{(i)}y^{(j)}+\sum_{i+j=k+2^{n-r}}\zeta_{s2^{n+1}}^{2^{n-r}}x^{(i)}y^{(j)}\right)\\&+\zeta_{s2^{n+1}}^{2^{n-r}-1}\sum_{i+j=2^{n-r}-1}x^{(i)}y^{(j)}.
\end{align*}
By letting \begin{align*}
b^{(k)}=\sum_{i+j=k}x^{(i)}y^{(j)}+\sum_{i+j=k+2^{n-r}}\zeta_{s2^{n+1}}^{2^{n-r}}x^{(i)}y^{(j)}
\end{align*} 
for $0 \leq k \leq 2^{n-r}-2$ and 
\begin{align*}
b^{(2^{n-r}-1)}=\sum_{i+j=2^{n-r}-1}x^{(i)}y^{(j)},
\end{align*} 
we have proven our claim. As in Theorem \ref{s2svp}, we have $\lambda_1(\mathcal{I})=\lambda_1(\mathfrak{c})$, as required.
\end{proof}
The following algorithm may be used to compute the shortest vector in $\mathcal{I}$.
\begin{algorithm}
\SetKwInOut{Input}{input}\SetKwInOut{Output}{output} \Input{An ideal $\mathcal{I}$.} \Output{A shortest vector in the corresponding ideal lattice.} \BlankLine 
\nl \For{$\overline{r}=1$ \KwTo $n$}{
\nl Compute a basis $(b^{(i)})_{0 \leq i < \phi(s2^{\overline{r}+1})}$ of the ideal lattice $\mathfrak{c}=\mathcal{I} \cap \mathcal{O}_K$ where $K=\mathbb{Q}(\zeta_{s2^{n+1}}^{2^{n-\overline{r}}})$. \\
\nl \If{$(\zeta_{s2^{n+1}}^jb^{(i)})_{0 \leq i < \phi(s2^{\overline{r}+1}), 0 \leq j \leq 2^{n-\overline{r}}}$ is exactly a basis of the ideal lattice $\mathcal{I}$}{
\nl Find a shortest vector $v$ in the $\phi(s2^{\overline{r}+1})$-dimensional lattice $\mathfrak{c}$.\\
\nl Output $v$.}
}
\caption{SVP algorithm for general ideal lattices of $\mathbb{Z}[\zeta_{s2^{n+1}}]$}
\end{algorithm}
\\
\subsection{The Cyclotomic Field $L=\mathbb{Q}(\zeta_{sp^{n+1}})$}
As before, $p$ is a positive, odd prime and $s$ is a positive integer such that $\gcd(s,p)=1$.
\begin{theorem}
Let $\mathcal{I}$ be a nonzero ideal of $\mathbb{Z}[\zeta_{sp^{n+1}}]$ with prime factorisation
\begin{align*}
    \mathcal{I}=\mathfrak{p}_1\mathfrak{p}_2\dots \mathfrak{p}_t,
\end{align*}
where $\mathfrak{p}_i=(f_i(\zeta_{sp^{n+1}}),\rho_i)$ for rational primes $\rho_i$ are (not necessarily distinct) prime ideals. If $\rho_i^{\phi(s)} \equiv a \mod p$ for some $\gcd(p,a)=1$, write $\rho_i^{\phi(s)}=m_ip^{A_i}+1$ and let $r=\max\{r_i\}$, where
\begin{align*}
    r_i=
    \begin{cases}
    \min\{A_i-1,n\}, \hspace{0.5mm} &\text{if} \hspace{1mm} \rho_i^{\phi(s)} \equiv a \mod p,\\
    n \hspace{0.5mm} &\text{if $\rho_i=p$}.
    \end{cases}
\end{align*}
Then SVP in the lattice generated by $\mathcal{I}$ can be solved via solving SVP in a $\phi(sp^{r+1})$-dimensional lattice.
\end{theorem}
\begin{proof}
If $r=n$ the theorem vacuously holds, so we assume otherwise. We may assume WLOG that $r=r_1$. Following the notation of Theorem \ref{spsvp}, we denote by
\begin{align*}
    G=\{\sigma_i: 1 \leq i \leq sp^{n+1}-1, \gcd(i,p)=\gcd(i,s)=1\}
\end{align*}
the Galois group of $L$, and consider the subgroup $H=H_1 \times H_2 \times \dots \times H_{p^{n-r-1}}$, where $H_k$ is the cyclic group generated by $\langle \sigma_{ksp^{r+1}+1}\rangle$, which is a subgroup of the decomposition group of every $\mathfrak{p}_i$, since $\sigma_{ksp^{r+1}+1}(\rho_i)=\rho_i$, $\sigma_{ksp^{r+1}+1}(f_i(\zeta_{sp^{n+1}}))=\sigma_{ksp^{r+1}+1}(g_i(\zeta_{sp^{r+1}}))=g_i(\zeta_{sp^{r+1}})=f_i(\zeta_{sp^{n+1}})$, where $g_i(x)$ is an irreducible factor of $\Phi_{sp^{A_i}}(x)$. As shown in Theorem \ref{spsvp}, the fixed field of $H$ is $K=\mathbb{Q}(\zeta_{sp^{n+1}}^{p^{n-r}})$, which has ring of integers $\mathcal{O}_K=\mathbb{Z}[\zeta_{sp^{n+1}}^{p^{n-r}}]$. Let $\mathfrak{c}= \mathcal{I} \cap \mathcal{O}_K$. We claim that for any $a \in \mathcal{I}$, there exist $a^{(k)} \in \mathfrak{c}$ for $0 \leq k \leq p^{n-r}-1$ such that
\begin{align*}
    a=\sum_{k=0}^{p^{n-r}-1}\zeta_{sp^{n+1}}^ka^{(k)}.
\end{align*}
We prove the claim via induction. When $t=1$, the claim holds by Theorem \ref{spsvp}, so we assume the claim holds for $t-1$. Letting $\overline{\mathcal{I}}=\mathfrak{p}_1\mathfrak{p}_2\dots \mathfrak{p}_{t-1}$, we have $\mathcal{I}=\mathfrak{p}_t\mathcal{I}$. It suffices to show that for any $xy$, $x \in \overline{\mathcal{I}}, y \in \mathfrak{p}_t$, there exist $b^{(k)} \in \mathcal{I} \cap \mathcal{O}_K$ for $0 \leq k \leq p^{n-r}-1$ such that $xy=\sum_{k=0}^{p^{n-r}-1}\zeta_{sp^{n+1}}^kb^{(k)}$. By the induction assumption, there exist $x^{(i)} \in \overline{\mathcal{I}} \cap \mathcal{O}_K$, $y^{(j)} \in \mathfrak{p}_t \cap \mathcal{O}_K$, $0 \leq i,j \leq p^{n-r}-1$ such that $x=\sum_{i=0}^{p^{n-r}-1}\zeta_{sp^{n+1}}^ix^{(i)}$ and $y=\sum_{j=0}^{p^{n-r}-1}\zeta_{sp^{n+1}}^jy^{(j)}$. Hence, we have
\begin{align*}
    xy&=\sum_{i,j=0}^{p^{n-r}-1}\zeta_{sp^{n+1}}^{i+j}x^{(i)}y^{(i)}
    \\&=\sum_{k=0}^{p^{n-r}-1}\zeta_{sp^{n+1}}^k\sum_{i+j=k}x^{(i)}y^{(j)}+\sum_{k=p^{n-r}}^{2 p^{n-r}-2}\zeta_{sp^{n+1}}^k
    \\&=\sum_{k=0}^{p^{n-r}-1}\zeta_{sp^{n+1}}^k\sum_{i+j=k}x^{(i)}y^{(j)}+\sum_{k=0}^{p^{n-r}-2}\zeta_{sp^{n+1}}^k\sum_{i+j=k+p^{n-r}}\zeta_{sp^{n+1}}^{p^{n-r}}x^{(i)}y^{(j)}
    \\&=\sum_{k=0}^{p^{n-r}-2}\zeta_{sp^{n+1}}^k\left(\sum_{i+j=k}x^{(i)}y^{(j)}+\sum_{i+j=k+p^{n-r}}\zeta_{sp^{n+1}}^{p^{n-r}}x^{(i)}y^{(j)}\right)\\&+\zeta_{sp^{n+1}}^{p^{n-r}-1}\sum_{i+j=p^{n-r}-1}x^{(i)}y^{(j)}.
\end{align*}
By letting \begin{align*}
b^{(k)}=\sum_{i+j=k}x^{(i)}y^{(j)}+\sum_{i+j=k+p^{n-r}}\zeta_{sp^{n+1}}^{p^{n-r}}x^{(i)}y^{(j)}
\end{align*} 
for $0 \leq k \leq p^{n-r}-2$ and 
\begin{align*}
b^{(p^{n-r}-1)}=\sum_{i+j=p^{n-r}-1}x^{(i)}y^{(j)},
\end{align*} 
we have proven our claim. As in Theorem \ref{spsvp}, we have $\lambda_1(\mathcal{I})=\lambda_1(\mathfrak{c})$, as required.
\end{proof}
Algorithm $4$ may be used to compute the shortest vector in $\mathcal{I}$.
\begin{algorithm}
\SetKwInOut{Input}{input}\SetKwInOut{Output}{output} \Input{An ideal $\mathcal{I}$.} \Output{A shortest vector in the corresponding ideal lattice.} \BlankLine 
\nl \For{$\overline{r}=1$ \KwTo $n$}{
\nl Compute a basis $(b^{(i)})_{0 \leq i < \phi(sp^{\overline{r}+1})}$ of the ideal lattice $\mathfrak{c}=\mathcal{I} \cap \mathcal{O}_K$ where $K=\mathbb{Q}(\zeta_{sp^{n+1}}^{2^{n-\overline{r}}})$. \\
\nl \If{$(\zeta_{sp^{n+1}}^jb^{(i)})_{0 \leq i < \phi(sp^{\overline{r}+1}), 0 \leq j \leq p^{n-\overline{r}}}$ is exactly a basis of the ideal lattice $\mathcal{I}$}{
\nl Find a shortest vector $v$ in the $\phi(sp^{\overline{r}+1})$-dimensional lattice $\mathfrak{c}$.\\
\nl Output $v$.}
}
\caption{SVP algorithm for general ideal lattices of $\mathbb{Z}[\zeta_{sp^{n+1}}]$}
\end{algorithm}
\section{Modules over Cyclotomic Rings}
Throughout this section, we use the same notation as in section 3.1. We let $L=\mathbb{Q}(\zeta_N)$ be the cyclotomic field of conductor $N$, where $N$ is of the form $sq^{n+1}$ for some positive integer $s$ and some $q$ where $q$ is $2$ or an odd prime, $\gcd(s,q)=1$. We take a module $M$ with pseudo-basis $\langle \mathcal{I}_k,\mathbf{b}_k \rangle_{k=1}^d$. As in section 3.1, we associate to each pseudo-ideal $\mathfrak{I}_k\mathbf{b}_k$ the decomposition group $\Delta_{\mathcal{I}_k}$ and the decomposition field $K_{\mathcal{I}_k}$. As we showed in Lemma \ref{separate}, we may alternatively represent the module $M$ using the pseudo-basis $\langle \alpha_k \mathcal{I}_k,\mathbf{b}_k^\prime \rangle_{k=1}^d$ where $\mathbf{b}_k^\prime \in K_{\mathcal{I}_k}^D$ and $\alpha_k\mathcal{I}_k$ are fractional ideals. Let $\mathcal{J}_k=Q_k\alpha_k\mathcal{I}_k$ for some rational integer $Q_k$ such that $\mathcal{J}_k$ is an ideal of $\mathcal{O}_L$. Denote by $\Delta_{\mathcal{J}_k}$ the decomposition group of $\mathcal{J}_k$ and by $r_k$ the maximum integer such that all the embeddings that fix $K_{k}=\mathbb{Q}(\zeta_{sq^{n+1}}^{q^{n-r_k}})$ also fix $K_{\mathcal{I}_j}$ for $1 \leq j \leq d$ (that is, they must form a subgroup of $\Delta_{\mathcal{I}_j}$), and such that (as in section 5) we may express $\mathcal{J}_k$ as
\begin{align*}
    \mathcal{J}_k=\bigoplus_{i=0}^{q^{n-r_k}-1}\mathfrak{c}_k\zeta_{sq^{n+1}}^i,
\end{align*}
where $\mathfrak{c}_k=\mathcal{J}_k \cap \mathcal{O}_{K_k}$. Then we may express $M$ by
\begin{align*}
    M=\bigoplus_{k=1}^d\mathcal{I}_k\mathbf{b}_k=\bigoplus_{k=1}^d\frac{1}{Q_k}\mathbf{b}_k^\prime\bigoplus_{i=0}^{q^{n-r_k}-1}\mathfrak{c}_k\zeta_{sq^{n+1}}^i.
\end{align*}
Since each $r_k$ are chosen so that the embeddings that fix the field $K_k$ also fix $K_{\mathcal{I}_j}$ for $1 \leq j \leq d$ and under the canonical embedding of ideal lattices the coefficients next to the roots of unity in the expression above may be treated as orthogonal components (as we have already shown), under the canonical embedding of $M$, the terms multiplied by the roots of unity in the expression above may be treated as orthogonal components, and so the shortest vector of the submodule
\begin{align*}
    \mathcal{M}=\bigoplus_{k=1}^d\frac{1}{Q_k}\mathfrak{c}_k\mathbf{b}_k^\prime
\end{align*}
is also the shortest vector in the module $M$. This simplifies SVP in module lattices, as the original module generated a lattice of dimension $d\phi(sq^{n+1})$ under the canonical embedding, whilst the submodule above generates a lattice of dimension at most $d\phi(sq^{r+1})$ where $r=\max_k\{r_k\} \leq n$. \\ The following algorithm explains how the SVP is implemented in module lattices over cyclotomic rings, and uses notation and conventions that we have discussed in this section.
\begin{algorithm}
\SetKwInOut{Input}{input}\SetKwInOut{Output}{output} \Input{A module $M$ with pseudo-basis $\langle \mathcal{I}_k,\mathbf{b}_k \rangle_{k=1}^d$.} \Output{A shortest vector in the corresponding module lattice.} \BlankLine 
\nl \For{$1 \leq i \leq d$}{
\nl Compute the decomposition group $\Delta_{\mathcal{I}_i}$ of the pseudo-ideal $\mathcal{I}_i\mathbf{b}_i$ and decomposition field $K_{\mathcal{I}_i}$. \\
\nl Set $x_i=\norm_{L/K_{\mathcal{I}_i}}(\mathcal{I}_i)$.\\
\nl \For{$\sigma \in \Delta_{\mathcal{I}_i}$}{
\nl Find $x_{\sigma,i} \in \mathcal{I}_i$ that satisfies $x_{\sigma,i}\mathbf{b}_i=\sigma(x_i\mathbf{b}_i)$}
\nl Set $\alpha_i=\left(\sum_{\sigma \in \Delta_{\mathcal{I}_i}} x_{i,\sigma}\right)^{-1}$, $\mathbf{b}_i^\prime=\alpha_i^{-1}\mathbf{b}_i$\\
\nl Set $Q_i$ to be smallest integer such that $Q_i\alpha_i\mathcal{I}_i \subseteq \mathcal{O}_L$, set $\mathcal{J}_i=Q_i\alpha_i\mathcal{I}_i$
}
\nl Compute the field $K=\mathbb{Q}(\zeta_{sq^{n+1}}^{q^{n-t}})$ such that $K$ is a subfield of the compositum of all $K_{\mathcal{I}_i}$, $1 \leq i \leq d$.\\
\nl \For{$1 \leq i \leq d$}{
\For{$1 \leq \overline{r}_i \leq n$}{
\nl Compute a basis $(b_i^{(j)})_{0 \leq j \leq \phi(sq^{\overline{r}_i+1})}$ of the ideal $\mathfrak{c}_i=\mathcal{J}_i \cap \mathcal{O}_{K_i}$ where $K_i=\mathbb{Q}(\zeta_{sq^{n+1}}^{q^{n-\overline{r}_i}})$ \\
\If{$(\zeta_{sq^{n+1}}^kb_i^{(j)})_{0 \leq j \leq \phi(sq^{\overline{r}_i+1}),0 \leq k \leq q^{n-\overline{r}_i}}$ is exactly a basis of $\mathcal{J}_i$ and $K_i \subseteq K$}{
Set $\mathfrak{c}_i=\mathcal{J}_i \cap K_i$ and break
}
}
}
\nl Find the shortest vector $\mathbf{v}$ in the lattice generated by the module $\mathcal{M}=\bigoplus_{i=1}^d \frac{1}{Q_i}\mathfrak{c}_i \mathbf{b}_i^{\prime}$. \\
\nl Output $\mathbf{v}$.
\caption{SVP algorithm for a module lattice $M$ over a cyclotomic ring.}
\end{algorithm}
\section{SVP Average-Case Hardness}
We fix a large $M$, and select a prime ideal uniformly randomly from the set
\begin{align*}
    \{\mathfrak{p} \hspace{2mm} \text{a prime ideal}: N(\mathfrak{p})<M\}.
\end{align*}
Our method of reduction to a $\phi(s2^{r+1})$- or $\phi(sp^{r+1})$-dimensional sublattice $\mathfrak{c}$ only works when $p$ does not split completely in $L$, or equivalently, when $N(\mathfrak{p})=\rho$. By Chebotarev's density theorem, the number of primes less than $M$ that split completely in $L=\mathbb{Q}(\zeta_m)$ is approximately $\frac{M}{\phi(m)\log(M)}$, and hence there are $\frac{M}{\log(M)}$ prime ideals lying above those primes for which our reduction method cannot be applied. Now, if our algorithm is to provide a reduction, we need $N(\mathfrak{p})=\rho^f<M$ for some positive integer $m$, and so the prime ideal must lie over a rational prime $\rho$ such that $\rho<\sqrt{M}$. Hence, there are at most $\sqrt{M}$ of such primes, and as such, there are at most:
\begin{itemize}
    \item $\phi(s)2^{n-1}\sqrt{M}$ when $L=\mathbb{Q}(\zeta_{s2^{n+1}})$ for odd, positive integer $s \geq 3$,
    \item $\phi(s)(p-1)p^{n-1}\sqrt{M}$ when $L=\mathbb{Q}(\zeta_{sp^{n+1}})$ for odd positive integer $s$ and odd prime $p$, $\gcd(s,p)=1$.
\end{itemize}
Then, for the relevant factor $\alpha(L) \sqrt{M}$ listed above, the density of easy instances for our algorithm is at most $\frac{\alpha(L) \log(M)}{\sqrt{M}}$, which goes to zero as $M$ tends to infinity for all considered fields.
\\ \textbf{General Ideals and Modules:} For any of the distributions covered, we can make similar assertions about the density of easy cases in general ideals and modules. If the probability of choosing an easily solvable prime ideal lattice in a given distribution is $P$, then given a general ideal $\mathcal{I}$ that has $g$ distinct factors, the probability that the ideal is easily solvable is $P^g$, since we require that all the prime ideal factors are individually easy cases.
\\
The module case is slightly trickier to assess, given that, as far as the authors are aware, the definition of a decomposition group for modules is novel. However, if we take a pseudo-ideal $\mathcal{I}_k\mathbf{b}_k$ and let $\mathbf{b}_k=(b_{1,k},\dots,b_{D,k})$ for some $b_{i,K} \in \mathcal{O}_L$, then the decomposition field $K_k$ of the pseudo-ideal must contain the compositum of the decomposition fields of the principal ideals $\mathcal{I}_kb_{i,k}$ for $1 \leq i \leq D$, and so the decomposition group $\Delta_k$ must contain the subgroup of $\text{Gal}(L/\mathbb{Q})$ that fixes the compositum of these fields. If we have a $d \times D$ matrix defining the module, and we say that each of these ideals factorises into at most $g$ distinct factors, the probability of picking a module for which SVP is easily solvable is at least $P^{dDg}$. However, it is worth noting that this is an extremely crude lower bound as we have used very little information about the module as a structure. It is possible that further research into the decomposition groups of modules could suggest that the probability of picking an easily solvable module is much higher.
\section{Concluding Remarks}
In this paper, we have successfully generalised methods pioneered by Pan et. al. in \cite{idealrandomprime}. First, we showed that a solution to Hermite SVP for a general ideal lattice may be yielded by solving the Hermite-SVP with a smaller factor in a subideal by exploiting the decomposition group of the ideal. Moreover, we showed that a similar argument may be made for module lattices defined over the ring of integers of a field that is Galois over $\mathbb{Q}$, and present two methods by which we may approach the module case. For ideals of cyclotomic rings, we generalised Pan et. al.'s results to construct an efficient SVP algorithm for ideals of cyclotomic rings of arbitrary conductor, and showed that there exists certain classes of ideal lattices whose structure is significantly weaker to our algorithm than others. Moreover, we proved that our method may be generalised to the case of modules over arbitrary cyclotomic rings, which may have consequences for the security of MLWE. However, unlike the case for ideal lattices, the problem remains open to identify classes of modules over cyclotomic rings which may be weaker under attack by our algorithm.

\newpage
\appendix
\section{Supplementary Material}
\subsection{Proof of Theorem 8}
\begin{proof}
Let $f(x)$ be any irreducible factor of $\Phi_{p^As}(x)$. We take $t$ in Lemma \ref{tmonic} to be $p^{n-A+1}$ and check the conditions of Lemma \ref{tmonic} term by term. By Lemmas \ref{factornumber} and \ref{order}, $f(x)$ has order $p^As$ and degree $m$, where $m$ is the least positive integer such that $q^m \equiv 1 \mod sp^A$. We claim that the prime divisor $p$ of $t$ divides $e=sp^A$ but not $\frac{q^m-1}{e}$. Clearly $p$ divides $sp^A$, so we only need to show the latter. We proceed via induction, and assume first that $A=1$. Note that $q^{\phi(s)}=1 \mod s$ (a consequence of Euler's theorem), and so $q^{\phi(s)m^{\prime}} \equiv 1 \mod sp$, where $m^{\prime}$ is the smallest integer so that $q^{\phi(s)m^{\prime}} \equiv 1 \mod p$, and so $m \mid \phi(s)m^{\prime}$. If a contradiction to our claim were to hold, then we would have
\begin{align*}
    q^{\phi(s)m^{\prime}}-1 \equiv 0 \mod p^2.
\end{align*}
The above may be represented in terms of its decomposition into cyclotomic polynomials:
\begin{align*}
    q^{\phi(s)m^{\prime}}-1 \equiv \prod_{d \mid m^{\prime}} \Phi_d(q^{\phi(s)}) \mod p^2.
\end{align*}
By Lemma \ref{factornumber}, each cyclotomic polynomial factors into distinct irreducible factors mod $p$. Since $q^{\phi(s)} \equiv a \mod p$, the above can only be zero if the factor $(q^{\phi(s)}-a)$ occurs in the factorisation of two distinct cyclotomic polynomials, which implies that there is at least one polynomial $\Phi_k(x)$ containing the factor $(x-a)$ such that $k<m^{\prime}$. Then we must have
\begin{align*}
    q^{\phi(s)k}-1 \equiv \prod_{d \mid k} \Phi_{d}(q^{\phi(s)}) \mod p \equiv 0 \mod p,
\end{align*}
which is a contradiction since $m^{\prime}$ was defined to be the minimum integer such that $q^{\phi(s)m^{\prime}} \equiv 1 \mod p$. Then we must have $p \nmid \frac{q^{m^{\prime}\phi(s)}-1}{p}$ and since $m \mid m^{\prime} \phi(s)$, $p \nmid \frac{q^m-1}{e}$. Now, assume that $p \nmid \frac{q^{m^{\prime}\phi(s)p^{A-1}}-1}{p^A}$ for some $A \geq 1$, and as before we let $m^{\prime}$ be the minimum integer such that $q^{\phi(s)m^{\prime}} \equiv 1 \mod p$, so we must have that $m \mid m^{\prime}\phi(s)p^{A-1}$. By Lemma \ref{pqdivide}, since $p^A \mid q^{\phi(s)m^{\prime}p^{A-1}}-1$ and $p \nmid \frac{q^{\phi(s)m^{\prime}p^{A-1}}-1}{p^A}$, we must have $p^{A+1} \mid q^{\phi(s)m^{\prime}p^{A}}-1$ and $p \nmid \frac{q^{\phi(s)m^{\prime}p^{A}}-1}{p^{A+1}}$, and since $m \mid \phi(s)m^{\prime}p^{A}$, it holds that $p \nmid \frac{q^{\phi(s)m}-1}{e}$ as required.
Moreover, $p$ is an odd prime so $4 \nmid t$, hence all the requirements in Lemma \ref{tmonic} have been fulfilled. Therefore it follows that $f(x^{p^{n-A+1}})$ is also irreducible over $\mathbb{F}_q$. By this method we can get $\frac{\phi(p^As)}{m}$ irreducible factors of $\Phi_{p^{n+1}s}(x)$. Now, suppose that $\Phi_{p^{n+1}s}(x)$ factors into $\frac{\phi(p^{n+1}s)}{M}= \frac{(p-1)p^n\phi(s)}{M}$ distinct factors, where $M$ is the least integer such that $q^M \equiv 1 \mod p^{n+1}s$. By Lemma \ref{pqdivide}, we must have $M=mp^{n-A+1}$ and hence
\begin{align*}
    \frac{\phi(p^{n+1}s)}{M}=\frac{(p-1)p^n\phi(s)}{mp^{n-A+1}}=\frac{(p-1)p^{A-1}\phi(s)}{m}.
\end{align*}
Hence in this way we may obtain all the irreducible factors of $\Phi_{p^{n+1}s}(x)$.
\end{proof}
\end{document}